%% file: main.tex
\documentclass{lmcs}
\pdfoutput=1

\usepackage{lastpage}
\lmcsdoi{16}{1}{7}
\lmcsheading{}{\pageref{LastPage}}{}{}%
{Jan.~31,~2019}{Jan.~31,~2020}{}

\newif\ifshort
\newif\iflong

\longtrue

\usepackage{latexsym,amsmath,amssymb,amsthm}
\usepackage{mathpartir}
\usepackage{stmaryrd}
\usepackage{stmaryrd}
\usepackage{tikz-cd}

\newcommand{\inlineoronline}[1]{\ifshort{\(#1\)}\else{\[#1\]}\fi}

\newcommand{\preciser}{\sqsubseteq}
\newcommand{\dynr}{\sqsubseteq}
\newcommand{\equiprecise}{\mathrel{\sqsupseteq\sqsubseteq}}
\newcommand{\equidyn}{\mathrel{\sqsupseteq\sqsubseteq}}

\newcommand{\vdynr}{\rotatebox[origin=c]{-90}{$\dynr$}}

\newcommand{\dyn}{{?}}
\newcommand{\obcast}[2]{\langle{#2}\Leftarrow{#1}\rangle}
\newcommand{\uarrow}{\mathrel{\rotatebox[origin=c]{-30}{$\leftarrowtail$}}}
\newcommand{\uarrowpr}{\mathrel{\rotatebox[origin=c]{-30}{$\hookleftarrow$}}}
\newcommand{\darrow}{\mathrel{\rotatebox[origin=c]{30}{$\twoheadleftarrow$}}}
\newcommand{\upcast}[2]{\langle{#2}\uarrow{#1}\rangle}
\newcommand{\upcastpr}[2]{\langle{#2}\uarrowpr{#1}\rangle}
\newcommand{\dncast}[2]{\langle{#1}\darrow{#2}\rangle}

\newcommand{\type}{\,\text{type}}
\newcommand{\ctx}{\,\text{ctx}}
\newcommand{\dom}{\text{dom}}

\newcommand{\cat}{\mathbb}
\newcommand{\Ctx}{\text{Ctx}}
\newcommand{\Multi}{\text{Multi}}
\newcommand{\PTT}{\text{PTT}}
\newcommand{\PTTS}{\text{PttS} }
\newcommand{\CPM}{CPM }

\newcommand{\GTTS}{\text{GttS}}
\newcommand{\GTTC}{\text{GttC}}

\newcommand{\sem}[1]{\llbracket{#1}\rrbracket}
\newcommand{\interp}[1]{\llparenthesis{#1}\rrparenthesis}

\newcommand{\src}{s}
\newcommand{\tgt}{t}

\newcommand{\err}{\mho}
\newcommand{\id}{\text{id}}
\newcommand{\app}{\text{app}}
\newcommand{\pair}{\text{pair}}
\newcommand{\unit}{\text{unit}}
\newcommand{\tl}{\triangleleft}

\newcommand{\alt}{\,|\,}

\newcommand{\coreflection}{\text{CoReflect}}

\newcommand{\axiom}[2]{\inferrule*[Right=#1]{~}{#2}}

\newcommand{\funupcast}[6]{\lambda #6 : #3. \upcast{#2}{#4}{(#5 (\dncast{#1}{#3}{#6}))}}

\newcommand{\produpcast}[5]{(\upcast{#1}{#3}\pi_0{#5}, \upcast{#2}{#4}\pi_0{#5})}

\newcommand{\aagrantack}{This material is based upon work supported by
  the National Science Foundation under grant CCF-1453796. Any
  opinions, findings, and conclusions or recommendations expressed in
  this material are those of the authors and do not necessarily
  reflect the views of the National Science Foundation.}

\newcommand{\dlgrantack}{This research was partially supported by the
  United States Air Force Research Laboratory under agreement numbers
  FA-95501210370 and FA-95501510053. The U.S. Government is authorized
  to reproduce and distribute reprints for Governmental purposes
  notwithstanding any copyright notation thereon.  The views and
  conclusions contained herein are those of the authors and should not
  be interpreted as necessarily representing the official policies or
  endorsements, either expressed or implied, of the United States Air
  Force Research Laboratory, the U.S. Government or Carnegie Mellon
  University.}

\ifshort
\author{Max S. New}{Northeastern University, Boston, USA}{maxnew@ccs.neu.edu}{}{\aagrantack}
\author{Daniel R. Licata}{Wesleyan University, Middletown, USA}{dlicata@wesleyan.edu}{}{\dlgrantack}
\authorrunning{M.\,S. New and D.\,R. Licata}
\Copyright{Max S. New and Daniel R. Licata}
\fi

\begin{document}
\title{Call-by-Name Gradual Type Theory}

\iflong
\author[M.S.~New]{Max S. New}
\address{Northeastern University, Boston, USA}
\email{maxnew@ccs.neu.edu}
\thanks{\aagrantack}
\author[D.R.~Licata]{Daniel R. Licata}
\address{Wesleyan University, Middletown, USA}
\email{dlicata@wesleyan.edu}
\thanks{\dlgrantack}
\fi

\theoremstyle{plain}\newtheorem{theorem}[thm]{Theorem}
\theoremstyle{plain}\newtheorem{lemma}[thm]{Lemma}
\theoremstyle{definition}\newtheorem{definition}[thm]{Definition}



\ifshort
\relatedversion{An extended version of this article is available from the arXiv \cite{extended}.}
\acknowledgements{We thank Amal Ahmed for countless insightful discussions of this
  work.}
\fi

\bibliographystyle{plainurl}

\ifshort
\EventEditors{H\'{e}l\`{e}ne Kirchner}
\EventNoEds{1}
\EventLongTitle{3rd International Conference on Formal Structures for Computation and Deduction (FSCD 2018)}
\EventShortTitle{FSCD 2018}
\EventAcronym{FSCD}
\EventYear{2018}
\EventDate{July 9--12, 2018}
\EventLocation{Oxford, UK}
\EventLogo{}
\SeriesVolume{108}
\ArticleNo{25}
\nolinenumbers
\fi

\begin{abstract}
  We present \emph{gradual type theory}, a logic and type theory for
  call-by-name gradual typing.  We define the central constructions of
  gradual typing (the dynamic type, type casts and type error) in a
  novel way, by universal properties relative to new judgments for
  \emph{gradual type and term dynamism}.  These dynamism judgments
  build on prior work in blame calculi and on the ``gradual
  guarantee'' theorem of gradual typing.  Combined with the ordinary
  extensionality ($\eta$) principles that type theory provides, we
  show that most of the standard operational behavior of casts is
  \emph{uniquely determined} by the gradual guarantee.  This provides
  a semantic justification for the definitions of casts, and shows
  that non-standard definitions of casts must violate these
  principles.  Our type theory is the internal language of a certain
  class of preorder categories called \emph{equipments}.  We give a
  general construction of an equipment interpreting gradual type
  theory from a 2-category representing non-gradual types and
  programs.
  This construction is a semantic analogue of the interpretation of
  gradual typing using contracts, and use it to build some concrete
  domain-theoretic models of gradual typing.
\end{abstract}

\maketitle

\section{Introduction}

Gradually typed languages allow for static and dynamic programming
styles within the same language.
They are designed with twin goals of allowing easy interoperability
between static and dynamic portions of a codebase and facilitating a
smooth transition from dynamic to static typing.
This allows for the introduction of new typing features to legacy
languages and codebases without the enormous manual effort currently
necessary to migrate code from a dynamically typed language to a fully
statically typed language.
Gradual typing allows exploratory programming and prototyping to be
done in a forgiving, dynamically typed style, while later that code
can be typed to ease readability and refactoring.
Due to this appeal, there has been a great deal of research
on extending
gradual typing~\cite{tobin-hochstadt06,siek-taha06} to numerous language features such as parametric
polymorphism~\cite{ahmed08:paramseal,ahmed17,igarashipoly17,torolabradatanter19}, effect
tracking~\cite{gradeffects2014}, typestate~\cite{Wolff:2011:GT},
session types~\cite{igarashisession17}, refinement
types~\cite{lehmann17} and security types \cite{toro18:sec}.
Almost all work on gradual typing is based solely on operational
semantics, and recent work such as \cite{refined} has codified some of
the central design principles of gradual typing in an operational
setting.
In this paper, we are interested in complementing this operational
work with a type-theoretic and category-theoretic analysis of these
design principles. We believe this will improve our understanding of
gradually typed languages, particularly with respect to principles for
reasoning about program equivalence, and assist in designing and
evaluating new gradually typed languages because it gives criteria for
casts (that they form embedding-projection pairs) that imply
graduality.

One of the central design principles for gradual typing is
\emph{gradual type soundness}.
At its most general, this should mean that the types of the gradually
typed language provide the same type-based reasoning that one could
reasonably expect from a similar statically typed language, i.e. one
with runtime errors and general recursion.
While this has previously been defined using operational semantics and
a notion of \emph{blame}~\cite{wadler-findler09}, the idea of
soundness we consider here is that the types should provide the same
extensionality ($\eta$) principles as in a statically typed language.
This way, programmers can reason about the ``typed'' parts of gradual
programs in the same way as in a fully static language.  
This definition fits nicely with a category-theoretic perspective,
because the $\beta$ and $\eta$ principles correspond to definitions of
connectives by a \emph{universal property}.
\iflong{}

\fi
The second design principle is the \emph{gradual
  guarantee}~\cite{refined}, which we will refer to as
\emph{graduality} (by analogy with parametricity).
Informally, graduality of a language means that syntactic changes from
dynamic to static typing (or vice-versa) should result in simple,
predictable changes to the semantics of a term.
More specifically, if a portion of a program is made ``more
static''/``less dynamic'' then the new program should either have the
same behavior or result in a runtime type error.
Other observable behavior such as values produced, I/O actions
performed or termination should not be changed.
\iflong
In other words, a ``less dynamic'' program should expose ``less
information'': by making types more static, we limit the interface for
the program and thus hide behavior, replacing it with a runtime type
error.
Of course, limiting the interface is precisely what allows for the
typed reasoning principles that gradual type soundness requires.
\fi

In this paper, we codify these two principles of soundness and
graduality \emph{directly} into a logical syntax we dub (call-by-name)
\emph{Gradual Type Theory} (Section~\ref{sec:gtt}).
For graduality, we develop a logic of \emph{type and term dynamism}
that can be used to reason about the relationship between ``more
dynamic'' and ``less dynamic'' versions of a program, and to give
novel specifications/universal properties for the dynamic
type, type errors, and runtime type casts of a gradually typed
language.
These universal properties extend the judgmental approach to type theory
(see ~\cite{martinlof83sienna,pfenningdavies}) to the key features of
gradual typing.
For soundness, we assert $\beta$ and $\eta$ principles as axioms of term
dynamism, so that the logic models program behavior.
Furthermore, using the $\eta$ principles for types, we 
show that most of the operational rules of runtime casts of existing
(call-by-name) gradually typed languages are \emph{uniquely determined}
by these constraints of soundness and graduality
(Section~\ref{sec:theorems}).
\iflong
As an example application, uniqueness implies that a complicated
space-efficient contract enforcement scheme in a particular language
(e.g. as in \cite{siek-wadler10}) is equivalent to a standard wrapping
implementation, \emph{if} it satisfies soundness and graduality (which
might be separately provable by a logical relations argument).
Contrapositively, %
\else For example,\fi
uniqueness implies that any enforcement scheme in a specific gradually
typed language that is \emph{not} equivalent to the standard
``wrapping'' ones \emph{must} violate either soundness or graduality.
We have chosen call-by-name because it is a simple setting with the
necessary $\eta$ principles (for negative types) to illustrate our
technique.
We have in other work considered application to other evaluation
orders (\cite{newlicataahmed19}), which we discuss in more detail in
Section~\ref{sec:related}.

We give a sound and complete category theoretic semantics for gradual
type theory in terms of certain \emph{preorder categories} (double
categories where one direction is thin) (Section~\ref{sec:semantics}).
We show that the contract interpretation of gradual
typing~\cite{tobin-hochstadt08} can be understood as a tool
for constructing models (Section~\ref{sec:contract-translation}):
starting from some existing language/category $C$, we first implement
casts as suitable pairs of functions/morphisms in $C$, and then equip
every type with canonical casts to the dynamic type.
\iflong
Technically, the first step forms a double category from a 2-category by
interpreting vertical arrows as Galois insertions/coreflections, i.e.,
related pairs of an upcast and a downcast.
Second, from a suitable choice of dynamic type, we construct a
``vertical slice'' preorder category whose objects are vertical arrows
into the chosen dynamic type.
\fi
We apply this to construct some concrete models in domains
(Section~\ref{sec:models}).

Conceptually, gradual type theory is analogous to Moggi's \emph{monadic
  metalanguage}~\cite{moggi91}: it clarifies general principles present in
many different programming languages; it is the internal language of a
quite general class of category-theoretic structures; and, for a
specific language, a number of useful results can be proved all at once
by showing that a logical relation over it is a model of the type
theory.

\ifshort \vspace{-0.2in}\fi

\subsection{A logic of dynamism and casts}
Before proceeding to the technical details, we explain at a high level
how our type theory accounts for two key features of gradual typing:
graduality and casts.
The ``gradual guarantee'' as defined in \cite{refined} applies to a
surface language where runtime type casts are implicitly inserted based
on type annotations, but we will focus here on an analysis of fully
elaborated languages, where explicit casts have already been inserted
(so our work does not yet address gradual type checking).
The gradual guarantee as defined in~\cite{refined} makes use of a
\emph{syntactically less dynamic}\footnote{Throughout this work, we
  will use the words ``less than'' to mean ``less than or equal to''
  rather than ``strictly less than'', and similarly for other terms
  such as ``greater than'', ``more dynamic'', etc.}  ordering on
types:
the dynamic type (universal domain) $\dyn$ is the most
dynamic, and $A$ is less dynamic than $B$ if $B$ has the same structure
as $A$ but some sub-terms are replaced with $\dyn$ (for example, $A \to (B
\times C)$ is less dynamic than $\dyn \to (B \times \dyn)$, $\dyn \to \dyn$ and $\dyn$).
Intuitively, a less dynamic type constrains the behavior of the
program more, but consequently gives stronger reasoning principles.  
This notion is extended to closed well-typed \emph{terms} $t:A$ and $t' : A'$ with $A$ less
dynamic than $A'$: $t$
is \emph{syntactically less dynamic} than $t'$ if 
$t$ is obtained from $t'$ by replacing the input and
output type of each type cast with a less (or equally) dynamic type (in
\cite{refined} this was called ``precision'').
For example, if $\mathsf{add1} : \dyn \to \mathbb{N}$ and $\mathsf{true} :
\dyn$, then $\mathsf{add1}( (\dyn \Leftarrow \mathbb{N})(\mathbb{N}
\Leftarrow \dyn) \mathsf{true})$ (cast $\mathsf{true}$ from dynamic
to $\mathbb{N}$ and back, to assert it is a number) is syntactically
less dynamic than
$\mathsf{add1}( (\dyn \Leftarrow \dyn)(\dyn \Leftarrow
\dyn) \mathsf{true})$ (where both casts are the identity).
Then the gradual guarantee~\cite{refined} says that if $t$ is
syntactically less dynamic than $t'$, then $t$ is \emph{semantically
  less dynamic} than $t'$: either $t$ evaluates to a type error (in
which case $t'$ may do anything) or $t, t'$ have the same first-order
behavior (both diverge or both terminate with $t$ producing a less
dynamic value).
In the above example, the less dynamic term always errors (because
$\mathsf{true}$ fails the runtime $\mathbb{N}$ check), while the more
dynamic term only errors if $\mathsf{add1}$ uses its argument as a
number.
In contrast, a program that returns a different value than
$\mathsf{add1}(\mathsf{true})$ does will not be semantically less
dynamic than it.

The approach we take in this paper is to give a \emph{syntactic logic}
for the \emph{semantic} notion of one term being less dynamic than
another, with $\err$ (type error) the least element, and all term
constructors monotone.
We call this the \emph{term dynamism relation} $t \dynr t'$, and it includes not only
syntactic changes in type casts, as above, but also equational laws like
identity and composition for casts, and $\beta\eta$ rules---so $t
\dynr t'$ intuitively means that $t$ type-errors more than (or as
much as) $t'$, but is otherwise equal according to these equational laws.
A programming language that is a model of our type theory will therefore
be equipped with a semantic $t \sem{\dynr} t'$ relation validating
these rules, so $t \sem{\dynr} t'$ if $t$ type-errors more than $t'$
up to these equational and monotonicity laws.
In particular, making type cast annotations less dynamic will result
in related programs, and if $\sem{\dynr}$ is adequate (i.e., no
operationally distinguishable terms are order-equivalent), then this
implies the gradual guarantee~\cite{refined}.
Therefore, we say a model of gradual type theory ``satisfies
graduality'' in the same sense that we would say a language satisfies
parametricity.
We have developed operational models for CBV and CBPV languages that
interpret the semantic ordering here as a type of contextual
approximation \cite{newahmed18,newlicataahmed19}.

Next, we discuss the relationship between term dynamism and casts, the
most novel aspect of our theory.
Explicit casts in a gradually typed language are typically presented by
the syntactic form $(B \Leftarrow A) t$, and their semantics is either
defined by various operational reductions that inspect the structure of
$A$ and $B$, or by ``contract'' translations, which compile a language
with casts to another language, where the casts are implemented as
ordinary functions.
In both cases, the behavior of casts is defined by inspection on types
and part of the language definition, with little justification beyond
intuition and precedent.

In gradual type theory, on the other hand, the behavior of casts is
\emph{not} defined by inspection of types.
Rather, we use the new type and term dynamism judgments, which are
defined \emph{prior to} casts, to give a few simple and uniform rules
specifying casts in all types via a universal property (optimal
implementation of a specification).
Our methodology requires isolating two special subclasses of casts,
upcasts and downcasts.
An upcast goes from a ``more static'' to a ``more dynamic'' type--- for
instance $(\dyn \Leftarrow (A \to B))$ is an upcast from a function type
up to the dynamic type---whereas a downcast is the opposite, casting to
the more static type.
We represent the relationship ``$A$ is less dynamic than $B$'' by
a \emph{type dynamism} judgment $A \dynr B$ (which corresponds to
the ``na\"ive subtyping'' of \cite{wadler-findler09}).
In gradual type theory, the upcast $\upcast{A}{B}{}$ from $A$ to $B$ and the downcast
$\dncast{A}{B}{}$ from $B$ to $A$ can be formed whenever $A \dynr B$.
This leaves out certain casts like $(\dyn \times \mathbb{N}) \Leftarrow
(\mathbb{N} \times \dyn)$ where neither type is more dynamic than the
other.
However, as first recognized in \cite{henglein94:dynamic-typing}, these
casts are macro-expressible~\cite{felleisen90} as a composite of an
upcast to the dynamic type and then a downcast from it (define $(B
\Leftarrow A)t$ as the composite 
$\dncast{B}{\dyn}{\upcast{A}{\dyn}{t}}$).

A key insight is that we can give upcasts and downcasts dual
specifications using term dynamism, which say how the casts relate
programs to type dynamism.
If $A \dynr B$, then for any term $t : A$, the
upcast $\upcast{A}{B}{t} : B$ is the \emph{least} dynamic term of type
$B$ that is more dynamic than $t$.
In order-theoretic terms, $\upcast{A}{B}{t} : B$ is the
$\dynr$-meet of all terms $u : B$ with $t \dynr u$.
Downcasts have a dual interpretation as a $\dynr$-join.
Intuitively, this property means upcast $\upcast{A}{B}{t}$ behaves as
much as possible like $t$ itself, while supporting the additional
interface provided by expanding the type from $A$ to $B$.

This simple definition has powerful consequences that we explore in
Section~\ref{sec:theorems}, because it characterizes the upcasts and
downcasts up to program equivalence.
We show that standard implementations of casts are the \emph{unique}
implementations that satisfy $\beta, \eta$ and basic congruence
rules.
In fact, almost all of the standard operational rules of a simple
call-by-name gradually typed language are term-dynamism equivalences in
gradual type theory.  The exception is rules that rely on disjointness
of different type connectives (such as
$\dncast{\dyn\to\dyn}{\dyn}{\upcast{\dyn\times\dyn}{\dyn}{t}} \mapsto
\err$), which are independent, and can be added as axioms.

\subsection{Models of Gradual Typing}
\iflong In addition to axiomatizing graduality in gradual type theory,
we also consider categorical semantics and denotational models of the
theory.

We give a definition of a \emph{model} of gradual type theory in
\emph{cartesian preorder multi-categories} (CPMS), which are related
to categories internal to the category of preordered sets, i.e., sets
with a reflexive, transitive relation.
This presents a simple alternative, algebraic specification of type
and term dynamism.
A CPM is like a category where the sets of objects and arrows have the
structure of a preorder, and the source, target, identity and
composition functions are all monotone. However, as in type theory,
and in contrast to categories, arrows can have $0$ or more inputs
rather than exactly $1$.
The ordering on objects models type dynamism and the ordering on
terms models term dynamism, and the rest of the requirements
succinctly describe the relationship between those two notions.

To model the casts, we in addition need that for any two objects with
$A \dynr B$, there exist morphisms $A \to B$ and $B \to A$ that
model upcasts and downcasts.
In the category theory literature, this structure is
called an \emph{equipment} and we adapt existing constructions and
results from that work \cite{shulman2008framed}.

We then prove an \emph{initiality} theorem for gradual type theory
with respect to the category of models, extending the classical
correspondence between simply typed lambda calculus and cartesian
closed categories \cite{Lambek:1986} to gradual typing.
In logical terms, we show that gradual type theory is \emph{sound} and
\emph{complete} with respect to our notion of model.
Of course this is no accident, we used this notion of model as the
basis for our design of gradual type theory.
However, we prefer to present the syntax first, since it does not
require any knowledge of category theory to understand.

In addition to providing a different perspective on the structure of
type and term dynamism, the CPM semantics of gradual typing enables us
to systematically build models of gradual typing.
In particular, we present the ``contract interpretation'' of casts as
a semantic construction of a model of gradual typing from a
cartesian 2-category.
Furthermore we can decompose this construction into simple pieces.
First, we form a double category from a 2-category by interpreting
vertical arrows as Galois insertions/coreflections, i.e., related
pairs of an upcast and a downcast.
Second, from a suitable choice of dynamic type, we construct a
``vertical slice'' CPM whose objects are vertical arrows into the
chosen dynamic type.

Finally, we instantiate this construction with several concrete
models, and in doing so make a formal connection between gradual
typing and domain-theoretic interpretations of dynamic typing.
Such a connection has been folklore since the earliest days of higher
order contracts, and we make this precise by constructing models of
gradual type theory.
First, we give a simple first-order model that doesn't support
function types, but as a consequence is also elementary in that it
only requires a solution to a covariant fixed point equation.
Then, we show that Dana Scott's classical construction of a model of
types from retracts of a universal domain is an instance of our
contract construction, but is \emph{inadequate} for interpreting
gradual typing because it conflates type errors and nontermination.
Finally we show that a better model can be constructed by using a
category of ``ordered domains'' that in addition to the domain
ordering have a separate ``type error ordering'' that models term
dynamism.
\fi

\subsection{Overview}

The paper proceeds as follows
\begin{itemize}
\item In Section \ref{sec:gtt} we present the syntax of gradual type
  theory (GTT).
\item In Section \ref{sec:theorems}, we formulate and prove many
  theorems \emph{within} GTT, including many equivalences between
  casts.
\item In Section \ref{sec:semantics}, we define models of gradual type
  theory and prove that the syntax of GTT provides the initial model.
\item In Section \ref{sec:contract-translation}, we present a semantic
  formulation of the ``contract interpretation'' of gradual types as
  constructing a model of GTT from a suitable 2-category.
\item In Section \ref{sec:models}, we instantiate our contract
  interpretation with several concrete models, including classic
  domain-theoretic interpretations of dynamically typed lambda
  calculus.
\item Finally, in Section \ref{sec:related}, we discuss related work.
\end{itemize}

This article is an extended version of \cite{newlicata18}. The primary
differences are that we (1) expanded Section~\ref{sec:theorems} to
include proofs, illustrating the process of reasoning about gradually
typed programs in the logic of GTT, (2) expanded
Section~\ref{sec:semantics} to include additional details of the
initiality theorems, and (3) expanded Section~\ref{sec:models} to more
fully describe the construction of our denotational models.

While Sections~\ref{sec:semantics}, \ref{sec:contract-translation},
\ref{sec:models} heavily use category-theoretic and domain-theoretic
terminology and techniques, Sections~\ref{sec:gtt} and
\ref{sec:theorems} are self-contained presentations of the syntax and
derivable theorems of GTT that require no knowledge of category or
domain theory, so readers with an interest in gradual typing but not
these semantic techniques may prefer to focus on these sections on GTT
syntax.

\section{Gradual Type Theory}
\label{sec:gtt}

In this section, we present the rules of gradual type theory (GTT).
Gradual type theory presents the types, connectives and casts of
gradual typing in a modular, type-theoretic way: the dynamic type,
type error and casts are defined by rules using the \emph{judgmental
  structure} of the type theory, specifically the new judgments for
type and term dynamism which we add to the usual judgmental structure
of typed lambda calculus.
Since the judgmental structure is so important, we
first present a bare-bones type theory we call \emph{preorder type
  theory} (PTT) which only has base types.
We can then modularly define what it means for this theory to have a
dynamic type, type errors, casts, functions and products.
Then gradual type theory is defined to be preorder type theory with
all of these constructions.

\iflong\subsection{Preorder Type Theory}
\else\subparagraph*{Preorder Type Theory}\fi
\label{sec:gtt:ptt}

\begin{figure}
  \begin{small}
    \begin{mathpar}
    \iflong A \type \and \Gamma \ctx \and \fi
    \inferrule{A \type\and A' \type}{A \dynr A'}\qquad
    \inferrule{\Gamma \ctx\and \Gamma' \ctx}
              {\Phi : \Gamma \dynr \Gamma'}\qquad
    \inferrule{\Gamma\ctx \and A \type}{\Gamma \vdash t : A}\qquad
    \inferrule{\Phi : \Gamma \dynr \Gamma' \and
               A      \dynr A'  \\\\
               \Gamma  \vdash t  : A   \and
               \Gamma' \vdash t' : A'}
              {\Phi \vdash t \dynr t' : A \dynr A'}
    \end{mathpar}
  \end{small}
  \caption{Judgment Presuppositions of Preorder Type Theory}
  \label{fig:formation}
\end{figure}

Preorder type theory (PTT) has 6 judgments: types, contexts, type
dynamism, dynamism contexts, terms and term dynamism. Their
presuppositions (one is only allowed to make a judgment when these
conditions hold) are presented in Figure~\ref{fig:formation},
where $A \type$ and $\Gamma \ctx$ have no conditions. 
The types, contexts and terms (Figure~\ref{fig:type-term}) are
structured as a standard type theory.
Terms are treated as intrinsically typed with respect to a context and
an output type, contexts are ordered lists (this is important for our
definition of dynamism context below).
For bare preorder type theory, the only types are base types, and the
only terms are variables and applications of uninterpreted function
symbols \ifshort(whose rule we omit)\fi.
\ifshort
In the extended version \cite{extended}, we give a precise definition of a
\emph{signature} specifying valid base types, function symbols, and type
and term dynamism axioms.
\else
These are all given by a \emph{signature} $\Sigma =
(\Sigma_0,\Sigma_1,\Sigma_2,\Sigma_3)$, formally defined below in
Definition \ref{def:ptt-signature}.
\fi
A substitution $\gamma : \Delta \vdash \Gamma$ is defined as usual:
\begin{definition}
  A substitution $\gamma : \Delta \vdash \Gamma$ is a function which
  given a variable in the output context $x:A \in \Gamma$, produces a
  term of that type relative to the input context $\Delta \vdash
  \gamma(x) : A$.
\end{definition}
Our term language supports a notion of substitution where if $\gamma :
\Delta \vdash \Gamma$ and $\Gamma \vdash t : A$ then $\Delta \vdash
t[\gamma] : A$, defined in the standard way for each construction we
add.
Weakening, contraction and exchange are all special cases of the
admissible action of substitution.

\begin{figure}
  \begin{small}
    \begin{mathpar}
    \iflong
    \inferrule{X \in \Sigma_0}{X \type}\and
    \else
    \inferrule{X \text{ base type}}{X \type}\and
    \fi
    \inferrule{~}{\cdot \ctx}\and
    \inferrule{\Gamma \ctx \and A \type \and x \not\in \dom(\Gamma)}{\Gamma, x : A \ctx}\\
    \inferrule{~}
              {\Gamma, x : A, \Gamma' \vdash x : A}
    \iflong
    \and
    \inferrule{f \in \Sigma_2(A_0,\ldots; B)\and \Delta \vdash t_0 : A_0 \,\,\cdots}
              {\Delta \vdash f(t_0,\ldots) : B}\and
    \fi
  \end{mathpar}
  \end{small}
  \caption{Preorder Type Theory: Type and Term Structure}
  \label{fig:type-term}
\end{figure}

Next, we discuss the new judgments of type dynamism, dynamism
contexts, and term dynamism, shown in Figure~\ref{fig:type-ctx-prec}.
A type dynamism judgment $A \dynr B$ relates two well-formed types,
and is read as ``$A$ is less dynamic than $B$''.
In preorder type theory, the only rules are reflexivity
(\textsc{TyDyn-Refl}) and transitivity (\textsc{TyDyn-Trans}), which
make type dynamism a preorder, and any axioms from the signature
$\Sigma_1$ (\textsc{TyDyn-Ax}).

The remaining rules in Figure~\ref{fig:type-ctx-prec} define \emph{type
  dynamism contexts} $\Phi$, which are used in the definition of term
dynamism.
While terms are indexed by a type and a typing context, term dynamism
judgments $\Phi \vdash t \dynr t' : A \dynr A'$ are indexed by
two terms $\Gamma \vdash t : A$ and $\Gamma' \vdash t' : A'$,
such that $A \dynr A'$ ($A$ is less dynamic than $A'$)
and $\Gamma$ is less dynamic than $\Gamma'$.  
Thus, we require a judgment $\Phi : \Gamma \dynr \Gamma'$,
which lifts type dynamism to contexts pointwise (for any $x:A \in
\Gamma$, the corresponding $x':A' \in \Gamma'$ satisfies $A \dynr
A'$).
This uses the structure of $\Gamma$ and $\Gamma'$ as ordered lists: a
dynamism context $\Phi : \Gamma \dynr \Gamma'$ implies that $\Gamma$
and $\Gamma'$ have the same length and associates variables based on
their order in the context, so that $\Phi$ is uniquely determined by
$\Gamma$ and $\Gamma'$; %
\ifshort
this is sufficient because of an admissible exchange rule for terms.
\else
if we want to form a judgment $t \dynr t'$ where their contexts
are not aligned in this way, we can always use exchange on one of them
to align it with the other.
\fi
We notate dynamism contexts to evoke a logical relations
interpretation of term dynamism: under the conditions that
$x_0\dynr x_0' : A_0 \dynr A_0',\ldots$ then we have that $t
\dynr t' : B \dynr B'$.

\begin{figure}
  \begin{mathpar}
  \inferrule*[right=TyDyn-Refl]{~}{A \dynr A}\and
  \inferrule*[right=TyDyn-Trans]{A \dynr A' \and A' \dynr A''}{A \dynr A''}\and
  \inferrule*[right=TyDyn-Ax]{(A,B) \in \Sigma_1}{A\dynr B}\\
  \inferrule*{~}
             {\cdot : \cdot \dynr \cdot}\and
  \inferrule*{\Phi : \Gamma \dynr \Gamma' \and A \dynr A'}
             {(\Phi, x \dynr x' : A \dynr A') : \Gamma, x : A \dynr \Gamma', x' : A'}
  \end{mathpar}
  \caption{Type and Context Dynamism}
  \label{fig:type-ctx-prec}
\end{figure}

The term dynamism judgment admits constructions
(Figure~\ref{fig:term-dyn}) corresponding to both the structural rules
of terms and the preorder structure of type dynamism, beginning from
arbitrary term dynamism axioms \ifshort(see the extended version \cite{extended} for a
formal definition)\else(\textsc{TmDyn-Ax})\fi .
First, there is a rule (\textsc{TmDyn-Var}) that relates
variables.
Next there is a \emph{compositionality} rule (\textsc{TmDyn-Comp})
that allows us to prove dynamism judgments by breaking terms down into
components. This uses a notion of substitution dynamism $\Phi\vdash
\gamma \dynr \gamma' : \Psi$ which is the pointwise extension of term
dynamism to substitutions:
\begin{definition}
  Given $\Phi : \Gamma \dynr \Gamma'$, $\Psi : \Delta \dynr \Delta'$,
  $\gamma : \Delta \vdash \Gamma$ and $\gamma' : \Delta' \vdash
  \Gamma'$, then
  \[ \Psi \vdash \gamma \dynr \gamma' : \Phi \]
  is defined to hold when for every $x \dynr x' : A \dynr A' \in \Phi$,
  $\Psi \vdash \gamma(x) \dynr \gamma'(x') : A \dynr A'$
\end{definition}
Last, we add an appropriate form of reflexivity (\textsc{TmDyn-Refl}) and
transitivity (\textsc{TmDyn-Trans}) as rules, whose well-formedness depends on
the reflexivity and transitivity of type dynamism.
While the reflexivity rule is intuitive, the transitivity rule is
more complex.
Consider an example where $A \dynr A' \dynr A''$ and $B
\dynr B' \dynr B''$:
\[
  \inferrule{x \dynr x':A \dynr A' \vdash t \dynr t' : B \dynr B'
    \and
    x' \dynr x'' : A' \dynr A''  \vdash t' \dynr t'' : B' \dynr B''
  }{x \dynr x'' : A \dynr A'' \vdash t \dynr t'' : B \dynr B''}
\]
In a logical relations interpretation of term dynamism, we would have
relations $\dynr_{A,A'}$, $\dynr_{A',A''}$, $\dynr_{A,A''}$ and
similarly for the $B$'s, and the term dynamism judgment of the
conclusion would be interpreted as saying that for any $u
\dynr_{A,A''} u''$, $t[u/x] \dynr_{B,B''} t''[u''/x'']$.
However, we could only instantiate the premises of the judgment if we
could produce some middle $u'$ with $u \dynr_{A,A'} u'
\dynr_{A',A''} u''$.
In such models, a middle $u'$ \emph{always} exists, because an
implicit condition of the transitivity rule is that $\dynr_{A,A''}$
is the relation composite of $\dynr_{A,A'}$ and $\dynr_{A',A''}$
(the composite exists by type dynamism transitivity, and type dynamism
witnesses are unique in PTT (thin in the semantics)).
PTT itself does not give a term for this $u'$, but the upcasts and
downcasts in gradual type theory do (take it to be $\upcast{A}{A'}{u}$
or $\dncast{A'}{A''}{u''}$).

\begin{figure}
  \begin{small}
    \begin{mathpar}
    \inferrule*[right=TmDyn-Var]
        {x \dynr x' : A \dynr A' \in \Phi}
        {\Phi \vdash x \dynr x' : A \dynr A'}\and
    \inferrule*[right=TmDyn-Comp]
        {\Phi \vdash t \dynr t' : A \dynr A' \and \Psi \vdash \gamma \dynr \gamma' : \Phi}
        {\Psi \vdash t[\gamma] \dynr t'[\gamma'] : A \dynr A'}\and
    \inferrule*[right=TmDyn-Refl]
               {\Gamma \vdash t : A\and \Phi : \Gamma \dynr \Gamma}
               {\Phi \vdash t \dynr t : A \dynr A}\and
    \inferrule*[right=TmDyn-Trans]
               {\Phi : \Gamma \dynr \Gamma'  \vdash t \dynr t' : A \dynr A' \\\\
                \Phi' : \Gamma' \dynr \Gamma'' \vdash t' \dynr t'' : A' \dynr A''\\\\
                \Psi : \Gamma \dynr \Gamma''
               }
    {\Psi : \Gamma \dynr \Gamma'' \vdash t \dynr t'' : A \dynr A''}\and
    \iflong
    \inferrule*[right=TmDyn-Ax]{(t, t') \in \Sigma_3\and
      \Gamma \vdash t : A\and
      \Gamma'\vdash t' : A'\and
      \Phi : \Gamma \dynr \Gamma'\and
    }{\Phi \vdash t \dynr t' : A \dynr A'}
    \fi
  \end{mathpar}
  \end{small}
  \caption{Primitive Rules of Term Dynamism}
  \label{fig:term-dyn}
\end{figure}

\ifshort
Sometimes it is convenient to use the same variable name at the same
type in both $t$ and $t'$, so 
we sometimes write $x:A$ in a dynamism context for $x \dynr x :
A \dynr A$, and write $\Gamma$ for $x_i \dynr x_i : A_i
\dynr A_i$ for all $x_i:A_i$ in $\Gamma$.  Similarly, we write $A$
as the conclusion of a dynamism judgment for $A \dynr A$, so
$\Gamma \vdash t \dynr t' : A$ means
$\Gamma \dynr \Gamma \vdash t \dynr t' : A \dynr A$.  
\else
We also introduce some convenient syntactic sugar for term dynamism
contexts and term dynamism, but for maximum clarity we will not use
the sugar when introducing rules, only when it shortens proofs we
present in the theory.
Sometimes it is convenient to use the same variable name at the same
type in both $t$ and $t'$ and so in such a case we simply write $x :
A$, which, in a type dynamism context is just a macro for $x
\dynr x : A \dynr A$ using the reflexivity of type dynamism.
Then with this sugar, type contexts are a subset of type dynamism
contexts.
Similarly when $t$ and $t'$ have the same output type we write $\Phi
\vdash t \dynr t' : A$ rather than the tediously long $\Phi \vdash
t \dynr t' : A \dynr A$.
\fi

\iflong
\subsection{PTT Signatures}

While gradual type theory proves that most operational rules of
gradual typing are equivalences, some must be added as axioms.
Compare Moggi's monadic metalanguage \cite{moggi91}: since it is a
general theory of monads, it is not provable that an effect is
commutative, but we can add a commutativity axiom and prove additional
consequences.
Similarly, in our type theory it is not provable without adding
additional axioms that an upcast followed by its complementary
downcast is the identity, or that the function type and product type
are disjoint.
To allow such axioms, preorder type theory is formally a \emph{family}
of type theories parameterized by a \emph{signature}; the signature is
also needed for a precise categorical semantics, because it represents
the ``generating data'' of a specific model.  

The signatures for preorder type theory (and, below, gradual type
theory) package together all of the base types, uninterpreted function
symbols and type and term dynamism axioms we desire.
This is mutually defined with the definition of the type theory
itself, so that for instance we can add function symbols whose
codomain is a non-base type.
%

\begin{definition}[PTT Signature]
  \label{def:ptt-signature}
  The notion of preorder type theory signature (PTT signature) is
  built as follows
  \begin{enumerate}
  \item A $0$-PTT signature is a set, and elements are called
    \emph{base types}.
  \item For a $0$-PTT signature $\Sigma_0$, $PTT_0(\Sigma_0)$ is the
    set of types generated by that signature and the rules of preorder
    type theory.
  \item A $1$-PTT Signature relative to a $0$-PTT signature $\Sigma_0$
    is a subset of $PTT_0(\Sigma_0)^2$, and elements are called
    \emph{type dynamism axioms}.
  \item A $2$-PTT Signature relative to a $0$-PTT signature
    $\Sigma_0$ is a set $\Sigma_2$ with functions $\src :
    \Sigma_2 \to PTT_0(\Sigma_0)^*$ and $\tgt : \Sigma_2 \to
    \PTT_0(\Sigma_0)$, and whose elements are called \emph{function
      symbols}.\footnote{technically the dependency on $\Sigma_1$ is trivial here, but is needed when we extend to GTT signatures.}
    We define $\Sigma_2(A_0,\ldots;B) = \{ t \in \Sigma_2 \alt \src(t) = A_0,\ldots \wedge \tgt(t) = B \}$
  \item For $0,1,2$-PTT signatures $\Sigma_0,\Sigma_1,\Sigma_2$,
    define $PTT_1(\Sigma_0,\Sigma_1,\Sigma_2)$ to be the set of all
    terms in PTT generated from the signatures.
  \item A $3$-PTT Signature $\Sigma_3$ relative to $0,1,2$-signatures
    $\Sigma_0,\Sigma_1,\Sigma_2$ is a set
    \[ \Sigma_3 \subseteq PTT_1(\Sigma_0,\Sigma_1,\Sigma_2)^2 \]
    such that if $(t, t') \in
    \Sigma_3$ and $\Gamma \vdash t : A$ and $\Gamma' \vdash t' : A'$,
    then it is derivable using $\Sigma_0,\Sigma_1,\Sigma_2$ that
    $\Gamma \dynr \Gamma'$ and $A \dynr A'$. Elements of
    $\Sigma_3$ are called \emph{term dynamism axioms}.
  \item Finally a PTT signature is a tuple of $0,1,2,3$-PTT signatures
    $(\Sigma_0,\Sigma_1,\Sigma_2,\Sigma_3)$, each relative to the
    previous signatures.
  \end{enumerate}
\end{definition}
\fi

\iflong\subsection{Gradual Type Theory}
\else\subparagraph*{Gradual Type Theory} \fi
\label{sec:gtt:dyn-cast}

Preorder Type Theory gives us a simple foundation with which to build
Gradual Type Theory in a modular way: we can characterize different
aspects of gradual typing, such as a dynamic type, casts, and type
errors separately.  

\iflong\subsubsection{Casts}\fi
We start by defining upcasts and downcasts, 
using type and term dynamism in Figure~\ref{fig:casts}.  
Given that $A \dynr A'$, the upcast is a function from $A$ to
$A'$ such that for any $t : A$, $\upcast{A}{A'}{t}$ is the
\emph{least dynamic term of type $A'$ that is at least as dynamic as $t$}.
The \textsc{UR} rule can be thought of as the ``introduction rule'',
saying $\upcast{A}{A'}{x}$ is more
dynamic than $x$, and then \textsc{UL} is the ``elimination rule'',
saying that if some $x' : A'$ is more dynamic than $x:A$, then it
is more dynamic than $\upcast{A}{A'}{x}$ --- since $\upcast{A}{A'}{x}$
is the \emph{least} dynamic term with this property.
The rules for projections are dual, ensuring that for $x' : A'$,
$\dncast{A}{A'}{x'}$ is the most dynamic term of type $A$ that
is less dynamic than $x'$.
\iflong

\fi
In fact, combined with the \textsc{TmDyn-Trans} rule, we can show that
it has a slightly more general property: $\upcast{A}{A'}{x}$ is not just
less dynamic than any term of type $A'$ more dynamic than $x$, but is
less dynamic than any term of type $A'$ \emph{or higher}, i.e. of type
$A'' \sqsupseteq A'$.
\iflong
Indeed, it is often convenient to use the following sequent-calculus
style rules (everything in the conclusion is fully general, except for
one cast), which are derivable using the
\textsc{TmDyn-Trans} and \textsc{TmDyn-Comp}
\begin{mathpar}
  \inferrule*[right=UR(S)]
             { \Phi \vdash t \dynr t' : A \preciser A' \and A' \dynr A''}
             {\Phi \vdash t \dynr \upcast{A'}{A''}{t'} : A \dynr A''}\and
  \and
  \inferrule*[right=UL(S)]
  { \Phi \vdash t \dynr t'' : A \dynr A'' \\\\ A \dynr A' \and A' \dynr A''}
  { \Phi \vdash \upcast{A}{A'}{t} \dynr t'' : A' \dynr A''}

  \inferrule*[right=DL(S)]
  {\Phi \vdash t' \dynr t'' : A' \dynr A'' \and A \dynr A'}
  {\Phi \vdash \dncast{A}{A'}{t'} \dynr t'' : A \dynr A''}\and

  \inferrule*[right=DR(S)]
  {\Phi \vdash t \dynr t'' : A \dynr A'' \\\\ A \dynr A' \and A' \dynr A''
  }
  {\Phi \vdash t \dynr \dncast{A'}{A''}{t''} : A \dynr A'}\and
\end{mathpar}
In particular, the upcast is left-invertible, and the downcast is
right-invertible (which agrees with their status as left and right
adjoints discussed below).

Though when read ``top-down'' the \textsc{UL(S)} and
\textsc{DR(S)} rules have more side-conditions in them than the
\textsc{UR(S)},\textsc{DL(S)} rules, when read ``bottom-up'', they
require fewer assumptions.
That is, in \textsc{UL(S)} (and similarly \textsc{DR(S)}) if we know the
conclusion is \emph{well-formed}, as we would when constructing a
proof, then the assumption that $A \dynr A'$ follows from the fact
that the upcast $\upcast{A}{A'}t$ is well-formed, $A' \dynr A''$
follows from the typing, and finally $A \dynr A''$ follows by
transitivity.
On the other hand, in \textsc{UR(S)} (similarly \textsc{DL(S)}), for the
premise to be well-formed, we need to know that $A \dynr A'$ which
does \emph{not} follow from the well-formedness of the conclusion.
So from a proof-construction standpoint, we can always apply a rule
when we have an upcast on the left or a downcast on the right, but we
must check a side-condition when we have an upcast on the right or
downcast on the left.
\fi

As we will discuss in Section~\ref{sec:theorems}, these rules allow us
to prove that the pair of the upcast and downcast form a \emph{Galois
  connection} (adjunction), meaning $\upcast{A}{A'}{\dncast{A}{A'}{t}}
\dynr t$ and $t \dynr \dncast{A}{A'}{\upcast{A}{A'}{t}}$.
However in existing gradually typed languages, the casts satisfy the
stronger condition of being a \emph{Galois insertion}, in which the
left adjoint, the downcast, is a \emph{retract} of the upcast, meaning
$t \equidyn \dncast{A}{A'}{\upcast{A}{A'}{t}}$.  We can restrict to
Galois insertions by adding the \emph{retract axiom}
\textsc{Retract}. Most theorems of gradual type theory do not
require it, though this axiom is satisfied in all models of preorder
type theory in Section~\ref{sec:models}.
\begin{figure}
  \begin{mathpar}
  \inferrule
      {\Gamma \vdash t : A \and A \dynr A'}
      {\Gamma \vdash \upcast{A}{A'}{t} : A'}\and

  \inferrule
  {\Gamma \vdash t : A' \and A \dynr A'}
  {\Gamma \vdash \dncast{A}{A'}{t} : A}
  \\
  \inferrule*[right=UR]
             {A \dynr A'}
             {x \dynr x : A \dynr A \vdash x \dynr \upcast{A}{A'}{x} : A \dynr A'}

  \inferrule*[right=DL]
  {A \dynr A'}
  {x' \dynr x' : A' \dynr A' \vdash \dncast{A}{A'}{x'} \dynr x' : A \dynr A'}

  \inferrule*[right=UL]
  {A \dynr A'}
  {x \dynr x' : A \dynr A' \vdash \upcast{A}{A'}{x} \dynr x' : A' \dynr A'}

  \inferrule*[right=DR]
  {A \dynr A'}
  {x \dynr x' : A \dynr A' \vdash x \dynr \dncast{A}{A'}{x'} : A \dynr A}
  
  \inferrule*[right=Retract]
   {A \dynr A'}
   {x : A \dynr x : A \vdash \dncast{A}{A'}{\upcast{A}{A'}{x}} \dynr x : A}
   \\
   \inferrule*{~}{\dyn \type}\and
   \inferrule*[right=$\dyn$Top]{~}{A \dynr \dyn}\and
   \inferrule*{~}{\Gamma \vdash \err_A : A}\and
   \inferrule*[right=$\err$Bot]{\Phi : \Gamma \preciser \Gamma}{\Phi \vdash \err_A \dynr t : A}
  \end{mathpar}
\caption{Upcasts, Downcasts, Dynamic Type and Type Error}
\label{fig:casts}
\end{figure}



\iflong\subsubsection{Dynamic Type and Type Errors} \fi
The remaining rules in Figure~\ref{fig:casts} define the dynamic type
and type errors, which are also given a universal property in terms of
type and term dynamism.
The dynamic type is defined as the most dynamic type (\textsc{$\dyn$Top}).
The type error, written as $\err$, is defined by the fact that it is a
constant at every type $A$ that is a least element of $A$ (\textsc{$\err$Bot}).
By transitivity, this further implies that $\err_{A} \dynr t : A
\dynr A'$ for any $A' \sqsupseteq A$.

\iflong \subsubsection{Negative Connectives} \fi
Next we illustrate how simple negative types can be defined in
preorder type theory in Figure~\ref{fig:negative}.

The first portion of the figure presents the rules for function types.
First, we have a rule to say a function type is well-formed, and
\textsc{${\to}$Mon} states that function types are monotone in
\emph{both} arguments with respect to term dynamism, following
previous work on type dynamism
\cite{henglein94:dynamic-typing,wadler-findler09,refined}.
Because of this, type dynamism is sometimes referred to as ``na\"ive
subtyping''.
See Section~\ref{sec:models:coreflections} for a semantic
explanation of the meaning of \textsc{${\to}$Mon}.
Next we have standard typing rules for $\lambda$ and application, and
corresponding monotonicity rules \textsc{$\lambda$Mon} and
\textsc{AppMon} for term dynamism.
Finally, we have call-by-name $\beta$ and $\eta$ rules, which we
present as equi-dynamism: we write $\equidyn$ to mean a rule exists in
each direction.

Next, the presentation of product types is much the same:
well-formedness and monotonicity of the type constructor, standard
introduction and elimination rules with corresponding monotonicity
rules and finally call-by-name $\beta\eta$ rules.
Lastly, we have three rules for the unit type.
First, we have a well-formedness rule, the corresponding monotonicity rule
is unnecessary because it follows from reflexivity of type dynamism (\textsc{TyDyn-Refl}).
Next, we have introduction, whose monotonicity follows from
reflexivity of term dynamism (\textsc{TmDyn-Refl}).
Finally we include the unit type's $\eta$ law. There is no $\beta$ law
because there is no elimination rule.

Specifically, we present the unit type,
products and function types in Figure \ref{fig:negative}.
The type and term constructors are the same as those in the simply
typed $\lambda$-calculus.
Each type constructor extends type dynamism in the standard way
: every
connective is \emph{monotone} in every argument, including the
function type.
Due to the covariance of the function type, For term dynamism, we add two classes of rules.
First, there are congruence rules that ``extrude'' the
term constructor rules for the type, which 
are like a ``congruence of contextual approximation'' condition.
Next, the computational rules reflect the ordinary $\beta,\eta$
equivalences as equi-dynamism: 

\begin{figure}
  \begin{mathpar}
    \inferrule*{A \type \and B \type}{A \to B \type}\and
    \inferrule*[right=$\to$Mon]
               {A \dynr A' \and B \dynr B'}
               {A \to B \dynr A' \to B'}
    \\
    \inferrule*{\Gamma, x : A \vdash t : B}
              {\Gamma \vdash \lambda x : A . t : A \to B}
    \and
    \inferrule*
        {\Gamma \vdash t : A \to B \and
        \Gamma \vdash u : A}
        {\Gamma \vdash t\,u : B}
    \\
    \inferrule*[right=$\lambda$Mon]{\Phi, x \dynr x' : A \dynr A' \vdash t \dynr t' : B \dynr B'}
    {\Phi \vdash \lambda x : A . t \dynr \lambda x' : A' . t' : A \to B \dynr A' \to B'}

    \inferrule*[right=AppMon]
        {\Phi \vdash t \dynr t' : A \to B \dynr A' \to B' \\\\
          \Phi \vdash u \dynr u' : A \dynr A'
        }
    {\Phi \vdash t\,u \dynr t'\,u' : B \dynr B'}
    \\
    \inferrule*[right=$\to\beta$]
        {\Phi : \Gamma \dynr \Gamma\\\\
          \Gamma \vdash t : A \to B\and
          \Gamma \vdash u : A
        }
        {\Phi \vdash  (\lambda x : A. t) u \equidyn t[u/x] : B \dynr B}
    \and
    \inferrule*[right=$\to\eta$]
    {\Phi : \Gamma \dynr \Gamma \\\\
      \Gamma \vdash t : A \to B}
    {\Phi \vdash t \equidyn (\lambda x : A. t\,x) : A \to B
      \dynr A \to B}\\\\
    \inferrule*{A_1 \type \and A_2 \type}{A_1 \times A_2 \type}\and
    \inferrule*[right=$\times$Mon]{A_1 \dynr A_1' \and A_2 \dynr A_2'}
              {A_1 \times A_2 \dynr A_1' \times A_2'}\\
    \inferrule*{\Gamma \vdash t_1 : A_1 \and \Gamma \vdash t_2 : A_2}
               {\Gamma \vdash (t_1, t_2) : A_1 \times A_2}
    \and
    \inferrule*{\Gamma \vdash t : A_1 \times A_2 \and i \in {1,2}}
              {\Gamma \vdash \pi_i t : A_i}
    \\
    \inferrule*[right=PairMon]
    {\Phi \vdash t_1 \dynr t_1' : A_1 \dynr A_1' \and
      \Phi \vdash t_2 \dynr t_2' : A_2 \dynr A_2'}
    {\Phi \vdash (t_1, t_2) \dynr (t_1',t_2') : A_1 \times A_2 \dynr A_1' \times A_2'}

    \inferrule*[right=PrjMon]
    {\Phi \vdash t \dynr t' : A_1 \times A_2 \dynr A_1' \times A_2'}
    {\Phi \vdash \pi_i t \dynr \pi_i t' : A_i \dynr A_i'}
    \\
    \inferrule*[right=${\times}\beta$]{i \in \{1,2\}}{\Gamma \vdash \pi_i(t_1,t_2) \equidyn t_i : A_i}
    \and
    \inferrule*[right=${\times}\eta$]{~}{\Gamma \vdash t \equidyn (\pi_1 t, \pi_2 t) : A_1 \times A_2}
    \\
    \inferrule*{~}{1 \type}\and
    \inferrule*{~}{\Gamma \vdash () : 1} \and
    \inferrule*[right=$1\eta$]{~}{\Gamma \vdash t \equidyn () : 1}    
  \end{mathpar}
  \caption{Function, Product and Unit Types}
  \label{fig:negative}
\end{figure}

We call the accumulation of all of these connectives \emph{gradual type
  theory}.
\ifshort
In the extended version \cite{extended}, we define a
GTT signature, which gives axioms for base types, function symbols,
type dynamism, and term dynamism, which all may make use of the
dynamic type, casts, type error, function types and product types, in
addition to the rules of PTT.
\else
A gradual type theory signature is a PTT signature where each
declaration can additionally use the structure of gradual type theory:
\begin{definition}[GTT Signature]
  A GTT signature $(\Sigma_0,\Sigma_1,\Sigma_2,\Sigma_3)$ is a PTT
  signature, where each declaration may make use of the rules for dynamic type,
  casts, type error, functions, products and unit types, in addition to
  the rules of PTT.
\end{definition}
\fi

\ifshort \input{theorems-short}
\else \input{theorems-long}\fi

\section{Categorical Semantics}
\label{sec:semantics}

Next, we define what a category-theoretic model of preorder and
gradual type theory is, and prove that PTT/GTT are \emph{internal
  languages} of these classes of models by proving that they are left
adjoint functors from categories of signatures to categories of
models.
This alternative axiomatic description of PTT/GTT is a useful bridge
between the syntax and the concrete models presented in
Section~\ref{sec:models}.  

The models are based on a variant of \emph{preorder categories}, which are categories
internal to the category of preorders.\footnote{To avoid confusion,
these are not categories that happen to be preorders (thin categories)
and these are not categories \emph{enriched} in the category of
preorders, where the hom-sets between two objects are preordered, but
the set of objects is not.}
A preorder category is a category where the set of all objects and set
of all arrows are each equipped with a preorder (a reflexive,
transitive, but not necessarily anti-symmetric, relation).
\ifshort
Furthermore the source, target, identity and composition functions are
all \emph{monotone} with respect to these orderings.
\else
That is, rather than having merely a \emph{set} of
objects and \emph{set} of arrows, preorder categories have a
\emph{preordered set} of objects and \emph{preordered set} of arrows
and the relevant functions are all monotone with respect to these
orderings.  
\fi
A preorder category is equivalently a double category where
one direction of morphism is thin.
Intuitively, the preorder of objects represents types and type
dynamism, while the preorder of morphisms represents terms and term
dynamism, and we reuse the notation $\dynr$ for the orderings on
objects and morphisms.

\iflong
\begin{definition}[Preorder Category]
  A preorder category $\cat C$ consists of
  \begin{enumerate}
  \item A preorder of ``objects'' $\cat{C}_0$
  \item A preorder of ``arrows'' $\cat{C}_1$
  \item Monotone functions of ``source'' and ``target'' $s, t : \cat{C}_1
    \to \cat{C}_0$ and ``identity'' $i : \cat{C}_0 \to \cat{C}_1$
  \item A monotone composition function $\circ : \cat{C}_1
    \mathrel{\times_{\cat{C}_0}} \cat{C}_1 \to \cat{C}_1$ where
    $\cat{C}_1 \mathrel{\times_{\cat{C}_0}} \cat{C}_1$ is the pullback
    of the source and target maps, which is explicitly given by the set
    \[ \cat{C}_1 \mathrel{\times_{\cat{C}_0}} \cat{C}_1 = \{ (f, g) \in \cat{C}_1^2 \alt s f = t g \}. \]
    This composition must satisfy that for any $(f,g) \in \cat{C}_1
    \mathrel{\times_{\cat{C}_0}} \cat{C}_1$, we have
    $s(f\circ g) = s g$ and $t(f\circ g) = t f$.
  \item Unitality and associativity laws for composition: $f \circ
    i(A) = f$, $i(B) \circ f = f$ and $(f \circ g) \circ h = f \circ
    (g \circ h)$ whenever these are well-defined.
  \end{enumerate}  
\end{definition}

The algebra of composition in a preorder category can be viewed using
a simple form of 2-dimensional string diagrams, which we adapt from \cite{shulman2008framed}.
We will not use this visualization for proofs, but we present it
because they help to understand the somewhat complex presuppositions
of the term dynamism judgment in gradual type theory.
We draw the objects as points, and draw arrows horizontally, while
drawing ordering relationships between objects vertically.
Then the ordering relationship between arrows forms a kind of square,
where the necessary ordering relationships are expressed
geometrically.
For instance if $f : A \to B$ is less than $f' : A' \to B'$, then it must
be the case that $A \dynr A'$ and $B \dynr B'$ because source and target
are monotone functions.
In gradual type theory, this is part of the presuppositions of the
term dynamism judgment.
This somewhat complex relationship between domains and codomains is
succinctly expressed as the following square which says that $f \dynr
f'$:

\[\begin{tikzcd}
A \arrow[dd, "\vdynr", dashed] \arrow[rr, "f"] &             & B \arrow[dd, "\vdynr", dashed] \\
                                                    & \vdynr &                                     \\
A' \arrow[rr, "f'"]                                 &             & B'                                 
\end{tikzcd} \]

Then, the substitution rule of gradual type theory can be
visualized as horizontal concatenation of squares.
If we have $f \dynr f'$ and $g \dynr g'$ then $g \circ f
\dynr g' \circ f'$, which is visualized as
\[\begin{tikzcd}
A \arrow[rr, "f"] \arrow[dd, "\vdynr", dashed] &        & B \arrow[rr, "g"] \arrow[dd, "\vdynr"] &        & C \arrow[dd, "\vdynr"] \\
                                               & \vdynr &                                        & \vdynr &                        \\
A' \arrow[rr, "f'"]                            &        & B' \arrow[rr, "g'"]                    &        & C'                    
\end{tikzcd}\]
And the transitivity rule similarly corresponds to vertical
concatenation of squares.
So we visualize $f \dynr f' \dynr f''$ as
\[ \begin{tikzcd}
A \arrow[rr, "f"] \arrow[dd, "\vdynr", dashed]   &        & B \arrow[dd, "\vdynr", dashed]  \\
                                                 & \vdynr &                                 \\
A' \arrow[rr, "f'"] \arrow[dd, "\vdynr", dashed] &        & B' \arrow[dd, "\vdynr", dashed] \\
                                                 & \vdynr &                                 \\
A'' \arrow[rr, "f''"]                            &        & B''                            
\end{tikzcd}\]
\fi

While the axioms of a preorder category are \emph{similar to} the
judgmental structure of preorder type theory, in a preorder category,
morphisms have \emph{one} source object and one target object, whereas
in preorder type theory, terms have an entire \emph{context} of inputs
and one output.
This is a standard mismatch between categories and type theories, and is
often resolved by assuming that models have product types and using
categorical products to interpret the context~\cite{Lambek:1986}.
However, we will take a \emph{multicategorical} view, in
which our notion of model will axiomatize algebraically a notion of
morphism with many inputs.
Though for ordinary simple type theory the difference between the two
is a matter of taste, for preorder type theory the difference is
important when modeling term and specifically context dynamism.
If a context $\Gamma = x_0:A_0,\ldots$ is modeled as a product of its
objects $A_0 \times \cdots$, then $\Gamma$ should be less dynamic than
another $\Gamma' = x_0':A_0',\ldots$ just when their product is
$A_0\times \cdots \dynr A_0'\times \cdots$. However, in the syntax of
our type theory, $\Gamma \dynr \Gamma'$ holds only when they are
pointwise $\dynr$, i.e., $A_0 \dynr A_0', \ldots$ all hold.
Since we allow for type dynamism axioms, these two notions are
\emph{not} equivalent, for instance there might be an axiom $1 \times
1\dynr 1$, which would mean that $x:1,y:1 \dynr z:1$ would hold if we
interpreted contexts simply as their product.
Therefore the syntax of context dynamism would be \emph{incomplete} if
it were interpreted as ordering of the products of the objects.
Instead, we give a multicategorical definition in which the notion of
context dynamism in the model is also pointwise.
Specifically, we define a model of preorder type theory to be a
\emph{cartesian preorder multicategory}, to be defined shortly, which
is like a preorder category that does not necessarily have true
product \emph{objects}, but whose morphisms' source can be a
``virtual'' product of objects, i.e. a context.

We developed this notion of cartesian preorder multi\-category using
the theory of generalized multicategories \cite{crutwell-shulman}.
Briefly, there is a monad on the (double) category of preorder
categories whose algebras are preorder categories with cartesian
products and cartesian preorder multi\-categories are the generalized
multicategories with respect to this monad.
While we will not make use of this abstract formulation, it was quite
useful in ensuring we had a reasonable definition.
\ifshort
In the extended version \cite{extended}, we prove soundness and completeness of
PTT for CPMs.  
\fi

\begin{definition}[CPM]
  \ifshort
  A cartesian preorder multi\-category (CPM) $\cat{C}$
  consists of
  a preordered set of ``objects'' $\cat{C}_0$, a
  preordered set of ``multiarrows'' $\cat{C}_1$, monotone functions
  ``source'' $s : \cat{C}_1 \to \Ctx(\cat{C})_0$, ``target'' $t :
  \cat{C}_1 \to \cat{C}_0$, ``projection'' $x : \Ctx(\cat{C})_0\times
  \cat{C}_0 \times \Ctx(\cat{C})_0 \to \cat{C}_1$ and composition
  $\circ : \cat{C}_1 \times_{\Ctx(\cat{C})_0} \Ctx(\cat{C})_1 \to
  \cat{C}_1$. Here $\Ctx(\cat{C})_0$ is the set of lists of objects
  preordered pointwise, and a substitution $\gamma \in
  \Ctx(\cat{C})_1(\Gamma;B_1,\ldots,B_n)$ consists of a multiarrow
  $\gamma(i) \in \cat{C}_1(\Gamma;B_i)$ for each $i \in 1,\ldots,n$,
  also preordered pointwise, and with composition defined in the same
  way as syntactic substitutions.
  the extended version \cite{extended}).  \else

  We mutually define what constitutes a cartesian preorder
  multi\-category (CPM) $\cat C$ with an associated preorder category
  $\Ctx (\cat C)$ of ``contexts'' and ``substitutions''.
  \begin{itemize}
  \item A cartesian preorder multi\-category (CPM) $\cat C$ consists of
    \begin{enumerate}
    \item a preordered set of ``objects'' $\cat{C}_0$
    \item a preordered set of ``multiarrows'' $\cat{C}_1$
    \item Monotone functions ``source'' $s : \cat{C}_1 \to \Ctx(\cat{C})_0$ and
      ``target'' $t : \cat{C}_1 \to \cat{C}_0$. We define $\cat{C}_1(\Gamma;A) = \{ f \in \cat{C}_1 \alt s(f) = \Gamma \wedge t(f) = A \}$
    \item Monotone ``projection'' functions $x :
      \Ctx(\cat{C})_0\times \cat{C}_0\times \Ctx(\cat{C})_0 \to \cat{C}_1$ satisfying
      $s(x(\Gamma;A;\Delta)) = \Gamma,A,\Delta$ and
      $t(x(\Gamma;A;\Delta)) = A$.
      We will sometimes refer to $(x(\cdot;A;\cdot)) \in \cat{C}_1(A;A)$ as $\id_A$, the identity multiarrow.
    \item A monotone ``composition'' function $\circ : \cat{C}_1
      \times_{\Ctx(\cat{C})_0} \Ctx(\cat{C})_1 \to \cat{C}_1$ satisfying
      $s(f \circ \gamma) = s(\gamma)$ and $t(f \circ \gamma) = t(f)$
      where $\cat{C}_1 \times_{\Ctx(\cat{C})_0} \Ctx(\cat{C})_1$ denotes
      a pullback, which is explicitly given by
      \[ \cat{C}_1 \times_{\Ctx(\cat{C})_0} \Ctx(\cat{C})_1 = \{ (f, \gamma) \alt t(\gamma) = s(f) \} \]
      equipped with the pointwise ordering.
    \item Satisfying the ``Identity'' law: for any $f \in \cat{C}_1(\Gamma;A)$,
      \[ f \circ \id_{\Gamma} = f\]
    \item Satisfying the ``Projection'' law, that for every $\gamma
      \in \Ctx(\cat{C})_1(\Theta;\Gamma,A,\Delta)$,
      \[
      x(\Gamma;A;\Delta) \circ \gamma = \gamma(|\Gamma|)
      \]
    \item Satisfying the ``Associativity'' law: for any $f \in
      \cat{C}_1(\Gamma;A)$, $\gamma \in \Ctx(\cat{C})_1(\Delta;\Gamma)$, $\delta \in \Ctx(\cat{C})_1(\Theta;\Delta)$
      \[
      (f \circ \gamma) \circ \delta = f \circ (\gamma \circ \delta)
      \]
    \end{enumerate}
  \item $\Ctx(\cat C)$ is a preorder category where
    \begin{enumerate}
    \item Objects $\Gamma \in \Ctx(\cat{C})_0$ are lists of elements of
      $\cat{C}_0$ (called ``contexts'') with the pointwise ordering.
    \item We define $\Ctx(\cat{C})_1(\Delta;\Gamma)$, the set of
      morphisms with source $\Delta$ and target $\Gamma = A_0,\ldots$
      to consist of functions $\gamma$ that assign for every $i \in
      \{0,\ldots,|\Gamma| - 1\}$ a multiarrow $\gamma(i) \in
      \cat{C}_1(\Delta; A_i)$. Then $\Ctx(\cat{C})_1$ is the set of
      triples $(\gamma; \Delta; \Gamma)$ such that $\gamma \in
      \Ctx(\cat{C})_1(\Delta;\Gamma)$, equipped with the pointwise
      ordering.
    \item Given $\gamma \in \Ctx(\cat{C})_1(\Delta;\Gamma)$ and
      $\gamma' \in \Ctx(\cat{C})_1(\Gamma;A_1,\ldots,A_n)$, we define
      the composition
      \[ (\gamma' \circ \gamma)(i) = \gamma'(i) \circ \gamma\]
      where the latter $\circ$ is
      composition of a substitution with a multiarrow.
    \item For an object $\Gamma = A_0,\ldots,A_n$, the identity substitution
      is given by the pointwise projections
      morphism.
      \[ \id_{\Gamma} = (x(\cdot;A_0;A_1,\ldots,A_n), \ldots, x(A_0,\ldots,A_{n-1};A_n;\cdot))\]
    \end{enumerate}
  \end{itemize}
  \fi
\end{definition}

\iflong An intuition for the axioms of multiarrow composition with
substitutions are that they are precisely what is needed to make the definition of
identity and composition for $\Ctx(\cat{C})$ into a cartesian preorder
category where the cartesian product is given by context concatenation.

\subsection{Soundness, Completeness and Initiality for PTT}
Next, we present the soundness and completeness theorems of the
interpretation of preorder type theory in a cartesian preorder
multicategory.
Soundness informally means that any interpretation of the base types
and function symbols in a CPM that satisfy the axioms in a signature
can be extended to a compositional semantics in which all derivable
type and term dynamism theorems are true in the model.
To make this precise, we first need to define precisely what a valid
interpretation of a signature in a CPM is. This proceeds in stages.
First we define an interpretation of base types.
\begin{definition}[Interpretation of Base Types]
  An interpretation of a $0$-PTT signature $\Sigma_0$ in a CPM $\cat
  C$ is a function $\interp{\cdot}_0 : \Sigma_0 \to \cat{C}_0$
  assigning an object to each base type.
\end{definition}

Then we define when an interpretation of base types validates all type
dynamism axioms.
\begin{definition}[Interpretation of Type Dynamism Axioms]
  An interpretation $\interp{\cdot}_0$ of $\Sigma_0$ in $\cat C$
  extends to an interpretation of a $1$-signature $\Sigma_1$, if for
  every type dynamism axiom $(A,B) \in \Sigma_1$, $\interp{A}_0 \dynr \interp{B}_0$
\end{definition}

And we can easily prove that any interpretation of type dynamism
axioms extends to a \emph{sound} interpretation of type dynamism
proofs.
\begin{theorem}[Soundness of PTT Type Dynamism]
  If $\interp{\cdot}_0$ is an interpretation of $\Sigma_0,\Sigma_1$ in
  $\cat C$, then if $A \dynr A'$ is derivable then $\interp{A}_0 \dynr
  \interp{A'}_0$ in $\cat C$.
\end{theorem}
\begin{proof}
  By induction on proofs of $A \dynr A'$. Reflexivity and transitivity
  follows by the fact that $\dynr$ is a preorder and axioms follow by
  assumption that $\interp{\cdot}_0$ is an interpretation of
  $\Sigma_1$.
\end{proof}

Then we define an interpretation of function symbols that extends an
interpretation of base types and type dynamism axioms.
\begin{definition}[Interpretation of Function Symbols]
  If $\interp{\cdot}_0$ is an interpretation of $\Sigma_0,\Sigma_1$ in
  a CPM $\cat C$, and $\Sigma_2$ is $2$-PTT signature extending
  $\Sigma_0,\Sigma_1$ then an extension of $\interp{\cdot}_0$ to
  interpret function symbols is a function $\interp{\cdot}_2 :
  \Sigma_2 \to \cat{C}_1$ that respects typing in that
  $t(\interp{f}_2) = \interp{t(f)}_0$ and $s(\interp{f}_2) =
  \interp{s(f)}_0$, where the last formula means the pointwise
  application of $\interp{\cdot}_0$ to each element of the list of
  input types.
\end{definition}

Next, we define how to extend such an interpretation to an
interpretation of all PTT terms.
\begin{definition}[Interpretation of PTT Terms]
  \label{def:ptt-soundness}
  If $\interp{\cdot}_0,\interp{\cdot}_1$ are interpretations of
  $\Sigma_0,\Sigma_1,\Sigma_2$ in a CPM $\cat C$, we extend this to an
  interpretation $\sem{\cdot}_2$ of all PTT terms in $\cat C$ as
  follows.
  \begin{align*}
      \sem{\Gamma \vdash f(t_0,\ldots)}_2 &= \interp{f} \circ (\sem{t_0}_2,\ldots)\\
      \sem{\Gamma,x:A,\Delta \vdash x : A}_2 &= x(\interp{\Gamma}_0;\interp {A}_0;\interp{\Delta}_0)
  \end{align*}
  where we define $\interp{x_0:A_0,\ldots}_0 = \interp{A_0}_0,\ldots$.
  Note that this respects typing in that if $\Gamma \vdash t : A$,
  then $s(\sem{t}_2) = \interp{\Gamma}_0$ and $t(\sem{t}_2) =
  \interp{A}_0$.
\end{definition}

Finally, we define when an interpretation satisfies all \emph{term} dynamism
axioms.
\begin{definition}[Interpretation of GTT Term Dynamism Axioms]
  If $\interp{\cdot}_0,\interp{\cdot}_2$ form an interpretation of
  $\Sigma_0,\Sigma_1,\Sigma_2$ in $\cat C$ and $\Sigma_3$ is a $3$-PTT
  signature relative to $\Sigma_0,\Sigma_1,\Sigma_2$, then we say
  $\interp{\cdot}_0,\interp{\cdot}_1$ interpret $\Sigma_3$ if for
  every term dynamism axiom $(t, u) \in \Sigma_3$, $\sem{t}_2 \dynr
  \sem{u}_2$ holds in $\cat C$.
\end{definition}

Then we can show that this interpretation of term dynamism is also
\emph{sound} in that if all term dynamism axioms are valid, then all
derivable term dynamism statements are true in the model.
\begin{theorem}[Soundness of Term Dynamism]
  If $\interp{\cdot}_0,\interp{\cdot}_2$ form an interpretation of
  $\Sigma = (\Sigma_0,\Sigma_1,\Sigma_2,\Sigma_3)$ in a CPM $\cat C$,
  then for every $\Phi \vdash t \dynr t' : A \dynr A'$ provable in PTT
  from $\Sigma$, $\sem{t}_2 \dynr \sem{u}_2$ holds in $\cat C$.
\end{theorem}
\begin{proof}
  By induction on term dynamism derivations.
  \begin{enumerate}
  \item Variable rule follows by monotonicity of projections
  \item Composition rule follows by monotonicity of $\circ$.
  \item Reflexivity and transitivity follow because $\dynr$ is a
    preorder in $\cat C$.
  \item Axioms follow by assumption.
    \qedhere
  \end{enumerate}
\end{proof}

We summarize these in the following soundness theorem.
\begin{theorem}[Soundness of Preorder Type Theory]
  \label{def:ptt-sem}
  For any PTT signature $\Sigma$ and CPM $\cat{C}$ and interpretation
  $\interp{\cdot}_0,\interp{\cdot}_1$ of $\Sigma$ in $\cat C$,
  \begin{enumerate}
  \item For every PTT type $A$, $\sem{A}_0$ is an object of $\cat{C}$.
  \item If $A \dynr A'$ is provable in PTT, then $\sem{A}_0 \dynr
    \sem{A'}_0$ in $\cat C$.
  \item For every $\Gamma \vdash t : A$ in PTT, $\sem{t}_2$ is a
    multi-arrow in $\cat C$ with $s(\sem{t}_2) = \interp{\Gamma}_0$ and
    $t(\sem{t}_2) = \interp{A}_0$.
  \item If $\Phi \vdash t \dynr t' : A \dynr A'$ is provable in GTT,
    then $\sem{t}_2 \dynr \sem{t'}_2$ in $\cat C$.
  \end{enumerate}
\end{theorem}

Next, we show the \emph{completeness} of PTT with respect to this
semantics, which informally means that if a type or term dynamism
ordering is satisfied in every CPM that interprets a signature, then
the ordering is syntactically derivable.
We prove it in the standard method for categorical models, which is to
show that the \emph{syntax itself} presents a CPM where term and type
dynamism are given by derivability of syntactic term and type dynamism
proofs.
\begin{theorem}[Completeness of CPM Semantics] 
  \label{thm:ptt-completeness}
  Let $\Sigma$ be a PTT signature and let all syntax be relative to
  $\Sigma$.
  \begin{enumerate}
  \item For any two types $A,A'$, if $\interp{A}_0 \dynr \interp{A'}_0$
    for every interpretation $\interp{\cdot}_0,\interp{\cdot}_1$ of
    $\Sigma$, then $A \dynr A'$ is derivable in PTT.
  \item For any two terms $t,t'$, if $\sem{t} \dynr \sem{t'}$ for
    every interpretation $\interp{\cdot}_0,\interp{\cdot}_2$, then $t
    \dynr t'$ is derivable in PTT.
  \end{enumerate}
\end{theorem}
\begin{proof}
  We construct the CPM $\PTT(\Sigma)$ as follows:
  \begin{enumerate}
  \item The objects are the types generated by $\Sigma$.
  \item $A \dynr A'$ holds when $A \dynr A'$ is derivable.
  \item A term $\PTT(\Sigma)(A_0,\ldots;B)$ is a term
    $x_0:A_0,\ldots \vdash t : B$ for some variables
    $x_0,\ldots,$, quotiented by $\alpha$-renaming (but not
    reordering). Composition is given by substitution and
    identity/projection by variable usage.
  \item $t \dynr t'$ holds when $\Phi \vdash t \dynr t' : A
    \dynr A'$ for the unique $\Phi,A,A'$ making that well-formed.
  \end{enumerate}

  Proving this is a CPM involves the standard proofs of the
  associativity and unitality of substitution and an easy proof that
  substitution is monotone with respect to term dynamism.

  There is an obvious interpretation of $\Sigma$ in $\PTT(\Sigma)$
  that interprets every base type and function symbol as itself. Then
  if a type/term dynamism theorem is true in every model, it is in
  particular true for $\PTT(\Sigma)$, which means exactly that it is
  derivable.
\end{proof}

Together these theorems imply \emph{initiality} of Preorder Type
Theory in the category of cartesian preorder multicategories,
analogous to the classical theorem for cartesian closed categories and
typed $\lambda$-calculus \cite{Lambek:1986}. Essentially the
initiality theorem packages up soundness and completeness into a
compact abstract statement, and formalizes the way in which preorder
type theory is a canonical syntax for CPMs.

First, we define a category of signatures and signature translations.
\begin{definition}[{Category of PTT Signatures}]
  We define a translation of PTT signatures $i : \Sigma \to \Theta$
  consists of
  \begin{enumerate}
  \item a function translating base types $i_0 : \Sigma_0 \to \Theta_0$.
  \item and a function translating function symbols $i_2 : \Sigma_2
    \to \Theta_2$ that respects typing in that $t(i_2(f)) = i_0(t(f))$
    and $s(i_2(f)) = i_0(s(f))$ where $i_0(s(f))$ is the pointwise
    application of $i_0$.
  \end{enumerate}
  Such that axioms are translated to axioms in that
  \begin{enumerate}
  \item if $(X,Y) \in \Sigma_1$, then $(i_0(X),i_0(Y)) \in \Theta_1$
  \item if $(f,g) \in \Sigma_3$, then $(i_2(f),i_2(g)) \in \Theta_3$.
  \end{enumerate}

  Translations compose by composing the components and this
  makes a category \PTTS whose objects are PTT signatures and
  morphisms are translations.
\end{definition}

Then we define a category of CPMs and functors.
\begin{definition}[{Category of CPMs}]
  A functor of CPMs $F : \cat C \to \cat D$ consists of a monotone
  function on objects $F_0 : \cat C_0 \to \cat D_0$ and a monotone
  function on multiarrows $F_1 : \cat C_1 \to \cat D_1$ that preserves
  source and target in that $t(F_1 (f)) = F_0(t(f)), s(F_1 (f)) =
  F_0(s(f))$, preserves composition in that $F_1(f \circ \gamma) =
  F_1(f) \circ F_1(\gamma)$ and preserves projection in that
  $F_1(x(\Gamma;A;\Delta)) =
  x(F_1(\Gamma);F_1(A);F_1(\Delta))$. Note that we
  implicitly have lifted $F_0,F_1$ pointwise to apply to contexts and
  substitutions

  This makes a category \CPM whose objects are CPMs and morphisms are
  functors with obvious identity and composition.
\end{definition}

Next, we define a functor that takes a CPM and produces the ``complete
signature'': one that includes all types as base types, all
multi-arrows as function symbols, and all true dynamism facts as
axioms.
\begin{definition}[{Complete Signature}]
  We define the complete signature functor $U : \CPM \to \PTTS$ on
  objects as follows.
  \begin{enumerate}
  \item $(U(\cat C))_0 = C_0$
  \item $(U(\cat C))_1 = \{ (A,A') \in \cat{C}_1^2 \alt A \dynr A' \}$
  \item $(U(\cat C))_2 = \cat{C}_1$
  \item $(U(\cat C))_3 = \{ (f(x_0,\ldots), g(x_0,\ldots)) \alt f,g \in \cat{C}_1 \wedge f \dynr g \wedge |s(f)| = n \}$
  \end{enumerate}
  $U(\cat C)$ has as base types the objects of
  $\cat C$, type dynamism axioms all true orderings between objects
  in $\cat C$, function symbols all morphisms of $U(\cat C)$ and term
  dynamism axioms all true orderings of morphisms in $\cat C$. This
  clearly extends to a functor.
\end{definition}

Then we can prove that model construction is a left adjoint functor to
this forgetful functor, establishing that the construction in the
proof of the completeness theorem produces the ``minimal'' way to make
a CPM from a signature.
The connection to the soundness and completeness theorems can be seen
by the fact that an interpretation $\interp{\cdot}_0,\interp{\cdot}_1$
of a signature $\Sigma$ in a CPM $\cat C$ is equivalent to a
translation of signatures $\Sigma \to U(\cat C)$. The initiality
theorem then states that any such interpretation uniquely extends to a
morphism of CPMs $PTT(\Sigma) \to \cat C$.
\begin{theorem}[PTT is the Initial CPM generated by a Signagture]
  The syntax of preorder type theory $PTT : \PTTS \to \CPM$ is a left
  adjoint functor to $U$.
\end{theorem}
\begin{proof}
  We use the formulation of a left adjoint in terms of universal
  morphisms \cite{maclane:71}. Note that this definition does not
  presuppose that $PTT$ is a functor, but rather proves that it is a
  functor.
  \begin{enumerate}
  \item As shown in the completeness theorem
    \ref{thm:ptt-completeness}, $PTT(\Sigma)$ is a CPM.
  \item For any signature $\Sigma$, there is a ``universal'' translation of signatures
     $\eta_{\Sigma} : \Sigma \to U(PTT(\Sigma))$ that
    interprets every generator as ``itself'', i.e.
    \begin{align*}
      \eta(X \in \Sigma_0) &= X\\
      \eta(f \in \Sigma_2) &= f(x_0,\ldots,x_{n-1})
    \end{align*}
    where $n$ is the length of $s(f)$. Every ordering axiom is
    satisfied in $U(PTT(\Sigma))$ by definition of $U$.
  \item For every translation $\interp{\cdot} : \Sigma \to U{\cat C}$, there is a
    morphism of CPMs $\sem{\cdot} : PTT(\Sigma) \to {\cat C}$ given by
    setting $\sem{A}_0 = \interp{A}_0$ and defining $\sem{t}_2$ as in
    the proof of the soundness Theorem~\ref{def:ptt-soundness}. The
    interpretation is functorial because it satisfies
    \begin{align*}
      \sem{t[\gamma]}_2 &= \sem{t}_2 \circ \sem{\gamma}_2
    \end{align*}
    which follows by induction.
  \item Every translation $\interp{\cdot} : \Sigma \to U{\cat C}$ factors through the
    universal translation in that $\interp{\cdot} = U(\sem{\cdot}) \circ
    \eta_{\Sigma}$.
    First, for any base type $X \in \Sigma_0$, we have
    \[ \sem{\eta_{\Sigma}(X)}_0 = \sem{X}_0 = \interp{X}_0 \]
    Next, for any function symbol $f \in \Sigma_1$ with input types $A_0,\ldots$, we
    have
    \begin{align*}
      \sem{\eta_{\Sigma}(f)}_2 &= \sem{f(x_0,\ldots)}_2 \\
      &= \interp{f}_2 \circ (\sem{x_0}_2,\ldots)\\
      &= \interp{f}_2 \circ \id_{(\sem{A_0}_0,\ldots)}\\
      &= \interp{f}_2
    \end{align*}
  \item $\sem{\cdot}$ is the \emph{unique} morphism satisfying this
    factorization. To show this, let $F : PTT(\Sigma) \to \cat{C}$ be
    a CPM morphism satisfying $UF \circ \eta_\Sigma = i$. We need to
    show that $F = \sem{\cdot}$.  First, on objects, because of the
    factorization
    \[ F_0(X) = \interp{X}_0 = \sem{X}_0 \]
    Next, for morphisms we proceed by induction on PTT terms.
    \begin{enumerate}
    \item For a variable $\Gamma,x:A,\Delta \vdash x : A$,
      \begin{align*}
        F_1(x) &= x(F_0(\Gamma);F_0(A);F_0(\Delta))\tag{functoriality}\\
        &= x(\sem{\Gamma}_0;\sem{A}_0;\sem{\Delta}_0)\tag{factorization on objects}\\
        &= \sem{x}_2\tag{definition}
      \end{align*}
    \item For function applications
      \begin{align*}
        F_1(f(t_0,\ldots)) &= F_1(f) \circ (F_1(t_0),\ldots)\tag{functoriality}\\
        &= \interp{f}_2 \circ (F_1(t_0),\ldots) \tag{factorization}\\
        &= \interp{f}_2 \circ (\sem{t_0}_2,\ldots) \tag{induction}\\
        &= \sem{f(t_0,\ldots)}_2 \tag{definition}
      \end{align*}
    \end{enumerate}
  \end{enumerate}
\end{proof}

\fi
\ifshort \subparagraph*{Gradual Typing Structures}
\else \subsection{Gradual Typing Structures}\fi

Next, we describe the additional structure on a CPM to model full
gradual type theory: casts are modeled by the structure of an
\emph{equipment}~\cite{shulman2008framed}, a dynamic type by a
greatest object, and the type error by a least element of every
hom-set.
\ifshort
\begin{definition}[Gradual Structure on a CPM]
  \hfill
\begin{enumerate}
  \item
    A CPM $\cat{C}$ is an \emph{equipment}
if for every $A \dynr A'$, there exist morphisms $u_{A,A'} \in
\cat{C}(A,A')$ and $d_{A,A'} \in \cat{C}(A',A)$ such that
$u_{A,A'} \dynr \id_{A'}$ and $\id_A  \dynr u_{A,A'}$  
and $d_{A,A'} \dynr \id_{A'}$ and $\id_A \dynr d_{A,A'}$.
An equipment is \emph{coreflective} if also $d_{A,A'} \circ u_{A,A'} \dynr \id_A$.  
\item
  A greatest object in a CPM $\cat{C}$ is a greatest element of the
preorder of objects $\cat{C}_0$.
\item
  A CPM $\cat{C}$ \emph{has local bottoms} if every hom set
  $\cat{C}(A_0,\ldots;B)$ has a least element $\bot$
  and for every substitution $\gamma$ we have $\bot \circ \gamma
  \equidyn \bot$.
  \end{enumerate}
\end{definition}
\else
\begin{definition}[Upcasts, Downcasts in a CPM]
  In a CPM $\cat{C}$, if $A \dynr A'$, we define
  \begin{enumerate}
    \item A morphism $u : A \to A'$ is an \emph{upcast} for $A \dynr A'$
      if $u \dynr \id_{A'}$ and $\id_A \dynr u$.
    \item A morphism $d : A' \to A$ is a \emph{downcast} for $A \dynr
      A'$ if $d \dynr \id_{A'}$ and $\id_A \dynr d$.
  \end{enumerate}
\end{definition}

It will be useful for the completeness theorem that upcasts and
downcasts are unique up to order-equivalence, a semantic version of
Theorem \ref{thm:syntactic-casts-are-uniq}.

\begin{lemma}[Upcasts, Downcasts are Unique]
  \label{lem:casts-are-uniq}
  If $u, u'$ are both upcasts for $A \dynr A'$, then $u \equidyn u'$.
  Similarly if $d,d'$ are both downcasts for $A \dynr A'$, then $d
  \equidyn d'$.
\end{lemma}
\begin{proof}
  By duality it is sufficient to show the upcast case. By symmetry of
  the situation it is sufficient to show $u \dynr u'$. First, $u = u
  \circ \id_{A}$ and $u' = \id_{A'} \circ u'$. Since they are upcasts,
  we also have $u \dynr \id_{A'}$ and $\id_A \dynr u'$, so by
  monotonicity of composition we have $u = u \circ \id_{A} \dynr \id_{A'}
  \circ u' = u'$.
\end{proof}

\begin{definition}[Equipment \cite{shulman2008framed}]
A CPM $\cat{C}$ is an \emph{equipment} if for every $A \dynr A'$,
there exist upcasts and downcasts for $A \dynr A'$.  An equipment is
\emph{coreflective} if also $d_{A,A'} \circ u_{A,A'} \dynr \id_A$.
\end{definition}

\begin{definition}[Greatest Object]
  A greatest object $\top$ in a CPM $\cat{C}$ is a greatest element of
  the preorder of objects $\cat{C}_0$.
\end{definition}

\begin{definition}[Local Bottoms]
  A CPM $\cat{C}$ \emph{has local bottoms} if every hom set
  $\cat{C}(\Gamma;A)$ has a least element $\bot_{\Gamma;A}$
  and for every substitution $\gamma \iflong\in
  \Ctx(\cat{C})_1(\Delta;\Gamma)\fi$ we have $\bot_{\Gamma;A} \circ \gamma
  \equidyn \bot_{\Delta;A}$.
\end{definition}
\fi

\iflong
\subsection{Interpreting Negative Types}
\fi
Next, we define a \emph{cartesian closed CPM}, which will model
negative function and product types.  
\iflong
While we use the adjectives ``closed'' and ``cartesian'', the
structure exhibited here is unique up to canonical
\emph{isomorphism}, but \emph{not} unique up to
\emph{order-equivalence} (equi-dynamism).
Thus, there may be multiple order-inequivalent ways that a CPM can be closed or cartesian, so it is important that e.g. a
closed CPM is a CPM \emph{with} a \emph{choice} of
exponentials.
\fi
\ifshort
We present the definition for function types in detail; a definition of
a cartesian CPM (CPM with products) is in \cite{extended}.  A
\emph{cartesian closed CPM} is a CPM with a choice of both cartesian and
closed structure.
\fi
Since all of the concrete models we consider are strict, we take a
strict interpretation of naturality and $\beta\eta$, but this could
likely be weakened.
\begin{definition}[Closed CPM]
  A Closed CPM is a CPM $\cat{C}$ with a monotone function
  on objects $\to : \cat{C}_0^2 \to \cat{C}_0$ making for every pair of objects
  $A,B \in \cat{C}$ an ``exponential'' object $A \to B$ with a monotone
  function
  \inlineoronline{\lambda : \cat{C}(\Gamma,A;B) \to \cat{C}(\Gamma;A \to B)}
  that is \emph{natural}
  \ifshort in an appropriate sense,
  \else in that for any appropriate $\Gamma,\gamma,h$
        \inlineoronline{\lambda(h) \circ \gamma = \lambda(h \circ
        (\gamma,x(\Gamma;A;\cdot)))}
  \fi
  with a morphism
  \inlineoronline{\app \in \cat{C}(A\to B, A; B)}
  such that the function given by
  \inlineoronline{f \mapsto \app \circ (f \circ \textrm{wkn}, x(\Gamma;A;\cdot)) : \cat{C}(\Gamma;A \to B) \to \cat{C}(\Gamma,A;B)}
  is an inverse to $\lambda$ (all up to equality), where here
  $\textrm{wkn} \in \Ctx(\cat{C})_1(\Gamma,A;\Gamma)$ is the evident
  weakening substitution.
\end{definition}

\iflong
\begin{definition}[Cartesian CPM]
  A Cartesian CPM is a CPM $\cat{C}$ with a monotone
  function $\times : \cat{C}_0^2 \to \cat{C}_0$ and a chosen object $1 \in \cat{C}_0$
  with functions
  \[ \pair : \cat{C}(\Gamma;A) \times \cat{C}(\Gamma;B) \to \cat{C}(\Gamma;A \times B) \]
  \[ \unit : 1 \to \cat{C}(\Gamma; 1) \]
  that are natural in that for any $f,g,\gamma$
  \[ \pair(f,g) \circ \gamma = \pair(f \circ \gamma, g \circ \gamma) \]
  \[ \unit \circ \gamma = \unit \]
  and morphisms
  \[ \pi_1 : \cat{C}(A\times B; A)\quad \pi_2 : \cat{C}(A\times B; B) \]
  such that the function given by
  \[ f \mapsto (\pi_1 \circ f, \pi_2 \circ f) : \cat{C}(\Gamma ; A \times B) \to \cat{C}(\Gamma;A) \times \cat{C}(\Gamma;B)\]
  is an inverse to $\pair$.
\end{definition}
A \emph{cartesian closed CPM} is a CPM with a choice
of both cartesian and closed structure.
\fi
\ifshort In the extended version \cite{extended}, we prove the following, where a \emph{GTT
  category} is a CPM that is cartesian closed, a coreflective equipment and has a greatest object and local bottoms.  \fi
\iflong
\begin{definition}
A \emph{GTT category} is a CPM that is cartesian closed, a coreflective
equipment and has a greatest object and local bottoms.
\end{definition}
\fi

\subsection{Soundness, Completeness and Initiality of GTT}
Next, we extend the interpretation of preorder type theory to gradual
type theory. Note that here the interpretation of casts depends on the
soundness of type dynamism, since we need to know that dynamism has a
sound interpretation in order for the semantic upcasts and downcasts
to exist. So we proceed in stages: first an interpretation of base
types, then type dynamism, then terms and finally term dynamism.

\begin{definition}[Interpretation of Base Types]
  An interpretation of a $0$-GTT signature $\Sigma_0$ in a GTT
  category $\cat{C}$ is a function $\interp{\cdot}_0 : \Sigma_0 \to
  \cat{C}_0$ assigning an object to each base type.
\end{definition}

\begin{definition}[Interpretation of GTT Types, Contexts]
  Given an interpretation $\interp{\cdot}_0$ of $\Sigma_0$ in a GTT
  category $\cat C$, we extend this to an interpretation
  $\sem{\cdot}_0$ of all types generated from $\Sigma_0$ as follows:
  \begin{align*}
    \sem{X}_0 &= \interp{X}_0\\
      \sem{\dyn}_0 &= \top\\
      \sem{A \to B}_0 &= \sem{A}_0 \to \sem{B}_0\\
      \sem{A \times B}_0 &= \sem{A}_0 \times \sem{B}_0\\
      \sem{1}_0 &= 1
  \end{align*}

  We extend this also to an interpretation of contexts by defining
  \[ \sem{x_0:A_0,\ldots}_0 = \sem{A_0}_0,\ldots\]
\end{definition}

\begin{definition}[Interpretation of Type Dynamism Axioms]
  An interpretation $\interp{\cdot}_0$ of $\Sigma_0$ in $\cat{C}$
  extends to an interpretation of a $1$-GTT signature $\Sigma_1$ if
  for every type dynamism axiom $(A, B) \in \Sigma_1$, $\sem{A}_0
  \dynr \sem{B}_0$.
\end{definition}

\begin{theorem}[Soundness of GTT Type Dynamism]
  If $\interp{\cdot}_0$ is an interpretation of $\Sigma_0,\Sigma_1$ in
  $\cat C$, then if $A \dynr A'$ is provable in GTT from
  $\Sigma_0,\Sigma_1$, then $\sem{A}_0 \dynr \sem{A'}_0$
\end{theorem}
\begin{proof}
  By induction on type dynamism derivations.
  \begin{enumerate}
  \item The reflexivity and transitivity cases follow by the fact that
    $\dynr$ in $\cat C$ is a preorder.
  \item $\sem{A}_0 \dynr \sem{\dyn}_0$ holds because $\sem{\dyn}_0 =
    \top$, which is a greatest element.
  \item The function and product cases follow because exponentials and
    products in $\cat C$ are assumed monotone.
  \qedhere
  \end{enumerate}
\end{proof}

\begin{definition}[Interpretation of Function Symbols]
  If $\interp{\cdot}_0$ is an interpretation of $\Sigma_0,\Sigma_1$ in
  $\cat C$, and $\Sigma_2$ is $2$-GTT signature extending
  $\Sigma_0,\Sigma_1$ then an extension of $\interp{\cdot}_0$ to
  interpret function symbols is a function $\interp{\cdot}_2 :
  \Sigma_2 \to \cat{C}_1$ that respects typing in that
  $t(\interp{f}_2) = \sem{t(f)}_0$ and $s(\interp{f}_2) =
  \sem{s(f)}_0$.
\end{definition}

\begin{definition}[Interpretation of GTT Terms, Substitutions]
  If $\interp{\cdot}_0,\interp{\cdot}_2$ form an interpretation of
  $\Sigma_0,\Sigma_1,\Sigma_2$ in $\cat C$, then we extend this to an
  interpretation $\sem{\cdot}_2$ of all terms in GTT generated by
  $\Sigma_0,\Sigma_1,\Sigma_2$ as follows:
  \begin{align*}
    \sem{f(t_0,\ldots)}_2 &= \interp{f}_2 \circ (\sem{t_0}_2,\ldots)\\
    \sem{\Gamma,x:A,\Delta \vdash x : A}_2 &= x(\sem{\Gamma}_0;\sem{A}_0;\sem{\Delta}_0)\\
    \sem{\upcast{A}{A'}{t}}_2 &= u_{\sem{A}_0,{\sem{A'}_0}} \circ (\sem{t}_2)\\
    \sem{\dncast{A}{A'}{t}}_2 &= d_{\sem{A}_0,{\sem{A'}_0}} \circ (\sem{t}_2)\\
    \sem{\err}_2 &= \bot\\
    \sem{\lambda x:A . t}_2 &= \lambda (\sem{t}_2)\\
    \sem{t u}_2 &= \app \circ (\sem{t}_2, \sem{u}_2)\\
    \sem{(t_1,t_2)}_2 &= \pair(\sem{t_1}_2,\sem{t_2}_2)\\
    \sem{\pi_i t}_2 &= \pi_i \circ (\sem{t}_2) \\
    \sem{()}_2 &= \unit
  \end{align*}
  And we extend it to substitutions by defining $\sem{\gamma}_2(i) =
  \sem{\gamma(x_i)}_2$ for $\gamma : \Delta \vdash x_0:A_0,\ldots$.
\end{definition}

\begin{definition}[Interpretation of Term Dynamism Axioms]
  If $\interp{\cdot}_0,\interp{\cdot}_2$ form an interpretation of
  $\Sigma_0,\Sigma_1,\Sigma_2$ in $\cat C$ and $\Sigma_3$ is a $3$-GTT
  signature relative to $\Sigma_0,\Sigma_1,\Sigma_2$, then we say
  $\interp{\cdot}_0,\interp{\cdot}_1$ interpret $\Sigma_3$ if for
  every term dynamism axiom $(t, u) \in \Sigma_3$, $\sem{t}_2 \dynr \sem{u}_2$.
\end{definition}

\begin{theorem}[Soundness of GTT Term Dynamism]
  If $\interp{\cdot}_0,\interp{\cdot}_2$ form an interpretation of
  $\Sigma = (\Sigma_0,\Sigma_1,\Sigma_2,\Sigma_3)$, then for every
  $\Phi \vdash t \dynr t' : A \dynr A'$ provable in GTT from $\Sigma$,
  $\sem{t}_2 \dynr \sem{t'}_2$.
\end{theorem}
\begin{proof}
  By induction on term dynamism derivations.
  \begin{enumerate}
  \item The variable, composition, reflexivity, transitivity and axiom
    rules follow b the same argument as for PTT.
  \item The upcast and downcast rules follow by definition of an
    equipment.
  \item The type error rule $\err \dynr t$ follows by
    definition of a local bottom.
  \item For the types $\to,\times,1$, the congruence rules hold by
    monotonicity of the cartesian and closed structures, and the
    $\beta\eta$ by the equational laws for the closed structure.
    \qedhere
  \end{enumerate}
\end{proof}

We summarize these results in the following Soundness theorem
\begin{theorem}[Soundness of Gradual Type Theory]
  \label{def:gtt-sem}
  For any GTT signature $\Sigma$ and GTT category $\cat{C}$ and
  interpretation $\interp{\cdot}$ of $\Sigma$ in $\cat C$,
  \begin{enumerate}
  \item For every GTT type $A$, $\sem{A}_0$ is an object of $\cat{C}$.
  \item If $A \dynr A'$ is provable in GTT, then $\sem{A}_0 \dynr
    \sem{A'}_0$.
  \item For every $\Gamma \vdash t : A$ in GTT, $\sem{t}_2$ is a
    multi-arrow in $\cat C$ with $s(\sem{t}_2) = \sem{\Gamma}_0$ and
    $t(\sem{t}_2) = \sem{A}_0$ where $\sem{x_0:A_0,\ldots}_0 =
    \sem{A_0}_0\ldots$.
  \item If $\Phi \vdash t \dynr t' : A \dynr A'$ is provable in GTT,
    then $\sem{t}_2 \dynr \sem{t'}_2$.
  \end{enumerate}
\end{theorem}

Next, we have the \emph{completeness} theorem for GTT which says that
if a dynamism is true in every interpretation of a signature, then it
is provable in GTT.
\begin{theorem}[Completeness of GTT Category Semantics] \label{thm:gtt-compl}
  Let $\Sigma$ be a GTT signature.
  \begin{enumerate}
  \item For any types $A,B$ in GTT, if for every interpretation
    $\interp{\cdot} : \Sigma \to \cat{C}$, $\sem{A}_0\dynr \sem{B}_0$
    holds, then $A \dynr A'$ is provable in GTT.
  \item For any GTT contexts $\Phi : \Gamma \dynr \Gamma'$, types $A
    \dynr A'$, and terms $\Gamma \vdash t : A$ and $\Gamma' \vdash t'
    : A'$ generated from $\Sigma$, if for every interpretation
    $\sem{t}_0 \dynr \sem{t'}_0$, then $\Phi \vdash t \dynr t' : A
    \dynr A'$ is provable in GTT.
  \end{enumerate}
\end{theorem}
\begin{proof}
  As with the proof of Theorem~\ref{thm:ptt-completeness}, we show
  that the interpretation used in the soundness Theorem
  ~\ref{def:gtt-sem} constructs a GTT category, i.e. $A \dynr A'$
  holds precisely when $A \dynr A'$ is derivable and similarly for $t
  \dynr t'$.  The proof is routine as the rules correspond precisely
  to the semantic notions.
\end{proof}

\iflong Together these theorems imply that the syntax is initial: the
semantics given by Definition~\ref{def:gtt-sem} is the unique
extension making a morphism of GTT categories using the GTT category
structure of Theorem~\ref{thm:gtt-compl}.  To formalize this, we
define first a category of GTT signatures $\GTTS$ in analogy with the
category of PTT signatures $\PTTS$.  We define the category of GTT
categories to have as functors CPM functors that preserve the GTT
structure.
\begin{definition}[{Category of GTT Categories}]
  We define the category $\GTTC$ to have as objects GTT categories and
  as morphisms $F : {\cat C} \to {\cat D}$ functors $F$ that preserve all GTT
  structure on the nose in that it
  \begin{enumerate}
  \item preserves greatest objects $F(\top) = \top$
  \item preserves local bottoms $F(\bot) = \bot$
  \item preserves exponentials in that $F(A \to B) = FA \to FB$ and
    $F(\lambda(f)) = \lambda(Ff)$ and $F(\app) = \app$.
  \item preserves cartesian products in that $F(A \times B) = FA
    \times F B$ and $F(\pair(f,g)) = \pair(F(f),F(g))$ and $F(\pi_i) =
    \pi_i$ and $F(1) = 1$ and $F(\unit) = \unit$
  \end{enumerate}
\end{definition}

Note that we need not specify that the functors preserve the
upcasts/downcasts because any functor of CPMs automatically preserves
upcasts/downcasts up to $\equidyn$, as shown for double categories in
\cite{shulman2008framed}.
\begin{lemma}[CPM Functors preserve Upcasts/Downcasts]
  \label{lem:functors-preserve-casts}
  If $F : {\cat C} \to {\cat D}$ is a CPM functor, then
  \begin{enumerate}
  \item If $u$ is an upcast for $A \dynr A'$, then $F(u)$ is an
    upcast for $F(A) \dynr F(A')$.
  \item If $d$ is a downcast for $A \dynr A'$, then $F(d)$ is a
    downcast for $F(A) \dynr F(A')$.
  \end{enumerate}
\end{lemma}
\begin{proof}
  We show the upcast case, the downcast case is dual. We need to show
  $\id_{F(A)} \dynr F(u)$ and $F(u) \dynr \id_{F(A')}$.
  Since $u$ is an upcast, $\id_A \dynr u$. By monotonicity of $F$,
  $F(\id_A) \dynr F(u)$. By functoriality of $F$, $F(\id_A) = \id_{F
    A}$, so $\id_{F A} \dynr F(u)$. $F(u) \dynr \id_{F(A')}$ follows by
  the same argument.
\end{proof}

Then as before there is an obvious ``complete signature functor'' $U :
\GTTC \to \GTTS$, and the interpretation function is a left adjoint,
with the proof an extension of the PTT case. This is a slight
generalization of the usual notion of left adjoint, because the
uniqueness is up to order-equivalence on arrows $\equidyn$, since we
don't require a choice of interpretation of upcasts/downcasts.
\begin{theorem}[GTT is the Initial GTT Category generated by a Signature]
  The GTT syntax $GTT : \GTTS \to \GTTC$ is left adjoint to $U$.
\end{theorem}
\begin{proof}
  Again we use the universal morphism definition of a left adjoint.
  \begin{enumerate}
  \item As shown in Theorem \ref{thm:gtt-compl}, $GTT(\Sigma)$ is a
    GTT category.
  \item For any signature $\Sigma$, there is a ``universal
    interpretation'' $\eta_{\Sigma} : \Sigma \to U(GTT(\Sigma))$
    interpreting the axioms/function symbols as ``themselves'', as
    in the universal interpretation for PTT.
  \item For every interpretation $\interp{\cdot} : \Sigma \to U{\cat
    C}$, there is a functor of GTT categories $\sem{\cdot} :
    GTT(\Sigma) \to {\cat C}$ given by the construction in the
    soundness Theorem \ref{def:gtt-sem}. It is a
    morphism of GTT categories because it is compositional and
    preserves exponentials, products, etc.
  \item Every interpretation $\interp{\cdot}$ factors through the
    universal interpretation in that $\interp{\cdot} = U(\sem{\cdot})
    \circ \eta_{\Sigma}$. This follows by the same argument as the PTT
    case.
  \item $\sem{\cdot}$ is the unique such factorization. Suppose
    $F : PTT(\Sigma) \to {\cat C}$ is a GTT functor such that $\interp{\cdot} =
    UF$ on objects and $\interp{\cdot} \equidyn UF$ on morphisms. We show
    that $F_0 = \sem{\cdot}_0$ on objects and $F_1 \equidyn
    \sem{\cdot}_2$ on morphisms.
    First, by induction on types.
    \begin{enumerate}
    \item For base types, as before $F_0(X) = \interp{X}_0 = \sem{X}_0$
      because $F$ factors through $\interp{\cdot}$.
    \item For the dynamic type, $F_0(\dyn) = \top = \sem{\dyn}_0$
      because $F$ is a GTT functor.
    \item For function types $F_0(A \to B) = F_0(A) \to F_0(B)$
      because $F$ is a GTT functor, then by induction this is
      equivalent to $\sem{A}_0 \to \sem{B}_0 = \sem{A \to B}_0$. The
      product and unit cases are similar.
    \end{enumerate}
    Next, by induction on terms.
    \begin{enumerate}
    \item For variables and function symbols, the same inductive
      argument as the PTT case follows.
    \item For the type error $F_1(\err) = \bot = \sem{\bot}_2$
      because $F$ is a GTT functor.
    \item For the upcast, we need to show $F_1(\upcast{A}{A'}t)
      \equidyn \sem{\upcast{A}{A'}{t}}_2$.  By functoriality we know
      $F_1(\upcast{A}{A'}t) = F_1(\upcast{A}{A'}x) \circ (F_1(t))$
      and by definition of the interpretation function,
      $\sem{\upcast{A}{A'}{t}}_2 = u_{\sem{A}_0,\sem{A'}_0} \circ
      (\sem{t}_2)$ so by induction and monotonicity it is sufficient to
      show $F_1(\upcast{A}{A'}x) \equidyn u_{\sem{A}_0,\sem{A'}_0}$. By
      Lemma \ref{lem:casts-are-uniq}, it is sufficient to show that
      $F(\upcast{A}{A'}x)$ is an upcast, which follows from Lemma
      \ref{lem:functors-preserve-casts}. The downcast case follows by
      the same argument.
    \item For $\lambda x. t$ , we have $F_1(\lambda x. t) =
      \lambda(F_1(t))$, and by induction $\lambda(F_1(t)) \equidyn
      \lambda(\sem{t}_2) = \sem{\lambda x. t}_2$. The pair/unit
      introduction cases is similar.
    \item For application $t u$, we calculate
      \begin{align*}
        F_1(t u) &= F_1(x_1 x_2) \circ (F_1(t),F_1(u)) \tag{functoriality}\\
        &= \app \circ (F_1(t),F_1(u)) \tag{$F$ preserves application}\\
        &\equidyn \app (\sem{t}_2,\sem{u}_2) \tag{induction}\\
        &= \sem{t u}_2
      \end{align*}
      The projection cases are similar.
      \qedhere
    \end{enumerate}
  \end{enumerate}
\end{proof}
\fi

\section{Semantic Contract Interpretation}
\label{sec:contract-translation}

As a next step towards constructing specific GTT categories, we define a
general \emph{contract construction} that provides a semantic account of
the ``contract interpretation'' of gradual typing, which models a
gradual type by a pair of casts.  The input to our contract construction
is a locally thin 2-category $\cat{C}$, whose objects and arrows should be
thought of as the types and terms of a programming language, and each
hom-set $\cat{C}(A,B)$ is ordered by an ``approximation ordering'', which is
used to define term dynamism in our eventual model.
We require each hom-set to have a least element (the type error), and
the category to be cartesian closed (function and product
types/contexts) in the strict sense of a 2-category whose underlying
category is cartesian closed and where $\lambda$, application, pairing
and projection are all functorial in 2-cells. The contract
construction then implements gradual typing using the morphisms of
the non-gradual ``programming language'' $\cat{C}$.

\ifshort\vspace{-0.125in}\fi
\subsection{Coreflections}
\label{sec:models:coreflections}

To build a GTT model from $\cat{C}$, we need to choose an interpretation of
type dynamism (the ordering on objects of the CPM) that
induces appropriate casts, which we know by
Theorem~\ref{thm:galois} must be Galois connections
that satisfy the retract axiom.
Such Galois connections are called Galois insertions (in order theory),
coreflections (in category theory) and embedding-projection pairs (in
domain theory).
In this paper we use the term coreflection since it is shortest.
While we presented the retract axiom earlier as an $\equidyn$, in all
of our models the semantics of the composition $\dncast{A}{A'} \circ
\upcast{A}{A'}$ is strictly equal to the identity so we will make a
model using ``strict'' coreflections because it is slightly simpler.
Since type dynamism judgments must induce a coreflection, we will construct
a model where the semantics of a type dynamism judgment $A \dynr A'$
\emph{is} a coreflection.
However, there can be many different coreflections between two objects
of our 2-category $\cat{C}$, so this first step of our construction does not
produce a preorder category, where type dynamism is an \emph{ordering},
but rather a \emph{double category}.
Double categories generalize preorder categories in the same way that
categories generalize preorders.
Double categories are categories internal to the category of small
categories, rather than the category of preorders.
The ordering on objects is generalized to proof-relevant
data specifying a second class of \emph{vertical morphisms}, and the
ordering on terms becomes a notion of 2-dimensional ``square'' between
morphisms.
\iflong
\begin{definition}[Double Category]
  A double category consists of
  \begin{enumerate}
  \item A category of objects and ``vertical'' arrows $\cat{C}_0$
  \item A category of ``horizontal'' arrows and 2-cells $\cat{C}_1$
  \item source, target and identity \emph{functors} with associativity
    and unitality axioms.
  \end{enumerate}
\end{definition}
\fi
In the model we build from $\cat{C}$, the vertical morphisms will model type
dynamism and be coreflections, while the \emph{horizontal} morphisms
of a preorder category will be arbitrary morphisms of $\cat{C}$ and model
terms.
We still require only double categories that are \emph{locally thin}, in
that there is at most one 2-cell filling in any square.
Thus, the first step of our contract construction can be summarized as
creating a double category that is an equipment
with the retract property, i.e. a double category modeling
upcasts and downcasts
\ifshort, a slight variation on a theorem in
\cite{shulman2008framed}
\fi.
\iflong
We can think of this as a model of a ``type dynamism proof-relevant''
system where there might be many different ways that $A \dynr A'$
(the next step in the construction will remedy this).
Then we get an interpretation of \emph{term dynamism} $\Phi \vdash t
\dynr t' : A \dynr A'$ as well, but
as squares whose sides are on the \emph{proofs} that $\Gamma \dynr \Gamma'$
and $A \dynr A'$ and the terms $t$ and $t'$.
Given specific coreflections $(u_A,d_A) : A \tl A'$ and $(u_B,d_B) : B
\tl B'$ in $\cat{C}$, then a 2-cell from $f : A \to B$ to $f' : A' \to B'$
along them should be thought of as a \emph{logical relatedness proof}.
For instance, in a set-theoretic model any coreflection induces a
\emph{relation} between its domain and codomain, so for instance we
have a relation $\dynr_{A,A'}$ that gives us a notion of when an
element of $A$ is less dynamic than an element of $A'$ by $x
\dynr_{A,A'} x'$ if $u_A(x) \dynr_{A'} x'$ or equivalently $x
\dynr_{A} d_{A}(x')$.
Then a 2-cell from $f$ to $g$ exists if for every $x \dynr_{A,A'}
x'$ then $f(x) \dynr_{B,B'} f'(x')$.
More formally, we can make the following construction, a slight
variation on a construction in \cite{shulman2008framed}.
\fi
\begin{definition}[Equipment of Coreflections \cite{shulman2008framed}]
  Given a 2-category $\cat{C}$ we construct a (double category) equipment
  $\coreflection(\cat{C})$ as follows.
  \ifshort
  Its object category has $\cat{C}_0$ as objects and coreflections in $\cat{C}$ as
  morphisms. Horizontal morphisms are given by morphisms in $\cat{C}$ and a
  2-cell from $f : A \to B$ to $f' : A' \to B'$ along $(u_A,d_A) : A
  \tl A'$ and $(u_B,d_B) : B \tl B'$ is a 2-cell in $\cat{C}$ from $u_B
  \circ f$ to $f' \circ u_A$ or equivalently from $f \circ d_A$ to
  $d_B \circ f'$. From a vertical arrow $(u,d)$, the upcast is $u$ and
  the downcast is $d$.
  \else
  \begin{enumerate}
  \item Its object category has $\cat{C}_0$ as objects and (strict) coreflections in
    $\cat{C}$ as morphisms, i.e., a vertical morphism $A \tl B$ is an
    adjoint pair of morphisms $u : A \to B$ and $d : B \to A$ where
    the unit is an equality: $d \circ u = \id$.  Composition
    of coreflections is covariant in the left adjoint and
    contravariant in the right adjoint: $(u,d) \circ (u', d') = (u
    \circ u', d' \circ d)$.
  \item Its arrow category $\coreflection(\cat{C})_1$ has morphisms of $\cat{C}$
    as objects and a 2-cell from $f : A \to B$ to $f' : A' \to B'$ is
    a triple of a coreflection $(u_A,d_A) : A \tl A'$, a coreflection
    $(u_B,d_B) : B \tl B'$ and a \emph{morphism of coreflections},
    i.e., a 2-cell in $\cat{C}$ $\alpha : u_B \circ f \Rightarrow f' \circ
    u_A$ which by a simple calculation can be equivalently presented
    as a morphism $\alpha' : f \circ d_A\Rightarrow d_B \circ f'$.
  \item The upcast from $(c_l, c_r)$ is $c_l$ and the downcast is
    $c_r$.
  \end{enumerate}
  \fi
\end{definition}

As is well-known in domain theory, covariant, contravariant and any
mixed-variance functor preserves coreflections, \cite{wand79ocat,
  smythplotkin}, so the product and exponential functors of $\cat{C}$
extend to be functorial also in vertical arrows.
This produces the classic ``wrapping'' construction familiar from
higher-order contracts \cite{findler-felleisen02}:\inlineoronline{(u,d) \to (u',d') = (d \to u', u \to d')}
\iflong
This construction preserves the structure from $\cat{C}$ that will
be needed to make a model of gradual type theory:
\begin{theorem}[Properties of $\coreflection(\cat{C})$]
  \hfill
  \begin{enumerate}
  \item If $\cat{C}$ is locally thin then so is $\coreflection(\cat{C})$.
  \item If a 2-category $\cat{C}$ has (pseudo) products and exponentials,
    then so does $\coreflection(\cat{C})$ because all functors preserve
    coreflections.
  \item If $C$ has local $\bot$s then so does $\coreflection(\cat{C})$.
  \end{enumerate}
\end{theorem}
We conjecture that this construction has a universal property: the
coreflection construction should be right adjoint to the forgetful
functor to 2-categories from the double category of coreflective
equipments. Intuitively this means that this is the ``maximal'' way
to add vertical arrows to a 2-category to make it a coreflective
equipment.

\fi

\subsection{Vertical Slice Category}
\label{sec:models:vertical-slice}

The double category $\coreflection{(\cat{C})}$ is not yet a model of gradual
type theory for two reasons.
First, gradual type theory requires a dynamic type: every type
should have a canonical coreflection into a specific type.
Second, type dynamism in GTT is \emph{proof-irrelevant}, because the
rules do not track different witnesses of $A \dynr A'$, but there
may be different coreflections from $A$ to $A'$.
It turns out that we can solve both problems at once by taking what we
call the ``vertical slice'' category\iflong\footnote{This definition
  of vertical slice category is not \emph{quite} the most natural from
  a higher categorical perspective because the horizontal arrows
  ignore the chosen object, but it is more useful for our
  purposes.} \fi over an object $D \in \coreflection(\cat{C})$ that is
rich enough to serve as a model of the dynamic type.
In $\coreflection(\cat{C}) / D$, the objects are not just an object $A$ of $\cat{C}$, but an object
\emph{with} a vertical morphism into $D$, in this case a coreflection
written $(u_A,d_A) : A \tl D$.\footnote{We do not write $A \dynr D$
  because coreflections are not a preorder.}  Thus, gradual types are
modeled as coreflections into the dynamic type, analogous to Scott's
``retracts of a universal domain'' \cite{scott76}.
Then a vertical arrow from $(u_A,d_A) : A \tl D$ to $(u_B,d_B) : B \tl
D$ is a coreflection $(u_{A,B},d_{A,B}) : A \tl B$ that
\emph{factorizes} $u_A = u_{B} \circ u_{A,B}$ and $d_A = d_{A,B}\circ
d_{B}$: this means the enforcement of $A$'s type can be thought of as
also enforcing $B$'s type.
Since upcasts are monomorphisms and downcasts are epimorphisms, this
factorization is \emph{unique} if it exists, so there is at most
one vertical arrow between any two objects of $\coreflection(\cat{C})/D$.
We could weaken this to monomorphism up to $\equidyn$ rather than
strict equality, in which case the factorization would only be
\emph{essentially unique}, i.e. any two factorizations would be
equivalent.
Since our models produce this stricter form of
monomorphism/epimorphism, we defer exploring a weak variant to future
work.
Further, the identity coreflection $(\id,\id) : D \tl D$ is a
vertically greatest element since any morphism is factorized by the
identity.

\begin{definition}[Vertical Slice Category]
  Given any double category $\cat{E}$ and an object $D \in
  \cat{E}$, we can construct a double category $\cat{E}/D$ by
  defining $(\cat{E}/D)_0$ to be the slice category $\cat{E}_0/D$,
  a horizontal morphism from $(c : A \tl D)$ to $(d : B \tl D)$ to be a
  horizontal morphism from $A$ to $B$ in $\cat{E}$, and the 2-cells
  are similarly inherited from $\cat{E}$.
\end{definition}

Next consider a cartesian closed structure on $\coreflection(\cat{C})/D$.
The action of $\to$ (respectively $\times,1$) on objects is
given by composition of the action in $\coreflection(\cat{C})$ $(u,d) \to
(u',d')$ with an \emph{arbitrary choice} of ``encoding'' of the ``most
dynamic function type'' $(u_{\to},d_{\to}) : (D \to D) \tl D$.
In most of the models we consider later, $D$ is a sum and this
coreflection simply projects out of the corresponding case, failing
otherwise.
This reflects the separation of the function contract into
``higher-order'' checking $(u,d) \to (u',d')$ and ``first-order tag''
checking $(u_{\to}, d_{\to})$ that has been observed in
implementations~\cite{henglein94:dynamic-typing}.

\iflong
We summarize the relevant results in the following theorem:
\begin{thm}[Vertical Slice Properties]
\hfill  \begin{enumerate}
  \item If $\cat{C}$ is an equipment, then so is $\cat{C}/D$.
  \item If $\cat{C}$ is cartesian, any pair of vertical morphisms
    $e_{\times} : D \times D \tl D$ and $e_{1} : 1 \tl D$ give $\cat{C}/D$
    the structure of a cartesian double category by defining $c \times
    d$ to be $e_{\times} \circ (c \times d)$ and inheriting the
    relevant morphisms from $\cat{C}$'s cartesian structure.
  \item If $\cat{C}$ is closed, any vertical morphism $e_{\to} : (D \to D)
    \tl D$ gives $\cat{C}/D$ the structure of a closed double category by
    defining $c \to d = e_{\to} \circ (c \to d)$.
  \item $\cat{C}/D$ is vertically thin (i.e., a preorder category) if and
    only if every vertical morphism in $\cat{C}$ is a monomorphism.
  \item If $\cat{C}$ has local $\bot$s then so does $\cat{C}/D$
  \end{enumerate}
\end{thm}

\subsection{Summary}
\fi

Finally, we construct a multi\-category $\Multi(\cat{C})$ from the
double category $\coreflection(\cat{C})/D$.  A multiarrow $A_0,
\ldots; B$ is given by a horizontal arrow from $A_1 \times ... 1$ to
$B$ in $\coreflection(\cat{C})/D$.  The ordering $A \preciser A'$ is
given by the vertical arrows $A \tl A'$ of $\coreflection(\cat{C})/D$
(i.e. coreflections), which is lifted pointwise to contexts by the
definition of a CPM.  The ordering $f : (A_0, \ldots; B)
\preciser g : (A_0', \ldots; B')$ is given by squares in
$\coreflection(\cat{C})/D$ (using the action/monotonicity of $\times$
on the pointwise orderings $A_i \preciser A_i'$ of the context).
\begin{definition}[Cartesian Preorder Multicategory from a Cartesian Preorder Category]
  If $\cat{C}$ is a cartesian preorder category, then we can construct a CPM
  category $\Multi(\cat{C})$ by
  \begin{enumerate}
  \item $\Multi(\cat{C})_0 = \cat{C}_0$
  \item $\Multi(\cat{C})(A_0,\ldots;B) = \cat{C}_1(A_0\times\cdots \times 1;B)$
  \end{enumerate}
\end{definition}
\begin{proof}
  This follows from a quite general result of \cite{crutwell-shulman}.  
\end{proof}
Combining these constructions, we produce:
\begin{theorem}[Contract Model of Gradual Typing]
  \label{thm:construct-gtt-model}
  If $\cat{C}$ is a locally thin cartesian closed 2-category with local
  $\bot$s, then for any object $D \in \cat{C}$ with chosen coreflections
  $c_{\to} : (D \to D) \triangleleft D$, $c_{\times} : (D \times D)
  \triangleleft D$, and $c_{1} : 1 \triangleleft D$, then
  $(\Multi({\coreflection(\cat{C})/id_d}), c_{\to},c_{\times},c_{1})$ is a GTT
  category.
\end{theorem}

\ifshort \input{models-short}
\else \input{models-long} \fi

\section{Related and Future Work}
\label{sec:related}

Since the original conference publication of this article
\cite{newlicata18}, we have developed~\cite{newlicataahmed19} a version of gradual type theory
based on call-by-push-value, which extends call-by-value and
call-by-name, and generalizes the type theory we have presented here.
The implementations of the dynamic
value type and dynamic computation type in \cite{newlicataahmed19} are
based on the two models in Section~\ref{sec:models-scott}.
We also developed operational models of CBPV gradual type theory,
based on an interpretation of term dynamism as a kind of contextual
approximation, following~\cite{newahmed18}.
While we did not develop a categorical semantics for
CBPV gradual type theory, there should be a similar
categorical semantics to the one presented here, by generalizing from
a preorder category to a (strong) adjunction of preorder categories.
In particular, our presentation of the operational models there is a
refinement of the contract construction we give here, and so should in
principle be also described by taking a kind of double category of
coreflections and vertical slice.  

\iflong
\subsection{Logic and Semantics of Dynamism}
\label{sec:related:logic}
\fi

Our logic and semantics of type and term dynamism builds on the
formulation introduced with the gradual guarantee in \cite{refined},
but the rules of our system differ in two key ways.
First, our system includes the $\beta, \eta$ equivalences as
equi-dynamism axioms, making term dynamism a more semantic notion.
Second, we only allow casts that are either upcasts or downcasts (as
defined by type dynamism), whereas their system allows for a more
liberal ``compatibility'' condition.
Accordingly our rules of dynamism for casts are slightly different,
but where it makes sense, the rules of the two systems are
interderivable.
\iflong
Modifying their cast rule to our syntax and ignoring any
compatibility constraint on the casted types, they have the following
two rules:
\begin{mathpar}
  \inferrule*[right=Cast-R]
  {\Phi \vdash t_1 \dynr t_2 : A_1 \dynr A_2 \and (A_1 \dynr B_2)}
  {\Phi \vdash t_1 \dynr \obcast{A_2}{B_2}{t_2} : A_1 \dynr B_2}

  \inferrule*[right=Cast-L]
  {\Phi \vdash t_1 \dynr t_2 : A_1 \dynr A_2 \and (B_1 \dynr A_2)}
  {\Phi \vdash \obcast{A_1}{B_1}{t_1} \dynr t_2 : B_1 \dynr A_2}
\end{mathpar}
Then we see that our four rules for upcast and downcasts are the
special case where the casts involved are upcasts or downcasts.
In the reverse direction, if we \emph{define} the ``oblique'' casts as
$\obcast{A}{B}t = \dncast{B}{\dyn}{\upcast{A}{\dyn}{t}}$, we can
derive their rules in \ref{fig:cast-rules}.
First, their cast right rule follows easily by applying our
sequent-style cast rules \textsc{DR(S)} and \textsc{UR(S)}.
The left rule takes slightly more work, using the retract axiom,
because we can't cast up to $\dyn$ on the left side because it might
be more dynamic than $A_2$.
Instead, we first prove that we can define the oblique cast
$\obcast{A}{B}{t}$ not just as the cast through $\dyn$, but also
through \emph{any} $C$ with $A,B\dynr C$ using the retract axiom and
composition of upcasts, downcasts.
Then we pick $C$ in the \textsc{Cast-L} case to be $A_2$, and then proof
proceeds dually to the \textsc{Cast-R} case.
\begin{figure}
  \begin{mathpar}
  \inferrule*[right=DR(S)]
  {\inferrule*[right=UR(S)]
    {\Phi \vdash t_1 \dynr t_2 : A_1 \dynr A_2}
    {\Phi \vdash t_1 \dynr {\upcast{A_2}{\dyn}{t_2}} : A_1 \dynr \dyn}}
  {\Phi \vdash t_1 \dynr \dncast{B_2}{\dyn}{\upcast{A_2}{\dyn}{t_2}} : A_1 \dynr B_2}\\

  \text{if } A,B \dynr C\\
  \dncast{B}{\dyn}{\upcast{A}{\dyn}{t}} \equidyn \dncast{B}{C}{\dncast{C}{\dyn}{\upcast{C}{\dyn}{{\upcast{A}{C}{t}}}}}\equidyn \dncast B C {\upcast A C t}

  \inferrule*[right=DL(S)]
  {\inferrule*[right=UL(S)]
   {\Phi \vdash t_1 \dynr t_2 : A_1 \dynr A_2}
   {\Phi \vdash \upcast{A_1}{A_2}{t_1} \dynr t_2 : A_2 \dynr A_2}}
  {\Phi \vdash \dncast{B_1}{A_2}{\upcast{A_1}{A_2}{t_1}} \dynr t_2 : B_1 \dynr A_2}
  \end{mathpar}
  \caption{\cite{refined} Cast Rules Derived}
  \label{fig:cast-rules}
\end{figure}
\fi
\iflong

As a relational logic with a sound and complete categorical semantics,
it has commonalities with logics for parametric polymorphism
\cite{plotkin1993logic}, and the categorical semantics in terms of
\emph{reflexive graph categories} which are like double categories
where vertical arrows lack composition \cite{ohearn95}.
In particular the System P logic presented in \cite{Dunphy:2002} is
similar to a ``dynamism proof-relevant'' version of preorder type
theory.
Additionally, the bifibration condition of
\cite{johann15bifibrational} is essentially the same as the
definition of an \emph{equipment}, but with a twist: in gradual typing
every contract induces an adjoint pair of terms, but there every term
induces an adjoint pair of \emph{relations}: the graph and
``cograph''.
Hopefully the similarity with parametric logics will be useful in
studying the combination of graduality with parametricity.
\fi

\iflong
\subsection{Contracts as Coreflections}
\fi

Our semantic model of contracts as coreflections has precedent in much
previous work, though we are the first to make precise the
relationship to gradual typing's notions of type and term dynamism.

\iflong
First, Dana Scott's seminal denotational work on models of the lambda
calculus is very similar to our vertical slice category: types are
modeled as retracts (or their associated idempotent) of a fixed
universal domain and morphisms are continuous functions of the
underlying domain (ignoring the universal domain).
Our treatment of type and term dynamism utilizes additional details
of this model, and the move from retracts to coreflections allows us
to give our specification for upcasts and downcasts.
Additionally, Scott's paper and later denotational work use
coreflections to solve mixed-variance domain equations \cite{scott76,
  wand79ocat, smythplotkin}.
The key reason is that one cannot construct a solution to $D \cong D
\to D$ as a limit or colimit because $\to$ has contravariant and
covariant arguments.
Instead, one moves to the category of coreflections where $\to$ is
covariant in both arguments.
Our coreflection model shows that this ``trick'' is also the reason
that the function type constructor is \emph{monotone} with respect to
type dynamism.
The double category setting allows us to better understand the
intertwined relationship between the categories of continuous maps and
coreflections and in this respect has much similarity to
\cite{pitts96}'s work, much of which could be fruitfully reframed in a
double categorical setting.
\fi

Henglein's work \cite{henglein94:dynamic-typing} on dynamic typing
defines casts that are retracts to the dynamic type, introduced the
upcast-followed-by-downcast factorization that we use here, and defines a
syntactic rewriting relation similar to our term dynamism rules.
Further they define a ``subtyping'' relation that is the same as type
dynamism and characterize it by a semantic property analogous to the
semantics of type dynamism in our contract model.
\ifshort
The upcast-downcast
factorization of an arbitrary cast is superficially similar to the work
on triple casts in \cite{siek-wadler10}, which collapse a sequence of
casts starting at $A$ and ending at $B$ into a downcast to $A \sqcap B$
followed by an upcast to $B$.  But note that this factorization is
opposite (downcast and then upcast), and the upcast-downcast
factorization requires only a dynamic type, while the converse
requires an appropriate middle type, similarly to \emph{image factorization}.
Moreover, \cite{AGT} shows that the correctness of
factorization through $A \sqcap B$ is not always possible.
\fi

Findler and Blume's work on contracts as \emph{pairs of projections}~\cite{findler-blume06}
is also similar.
There a contract is defined in an untyped language to be given by a
\emph{pair} of functions that divide enforcement of a type between a
``positive'' component that checks the term and a ``negative''
component that checks the continuation, naturally supporting a
definition of \emph{blame} when a contract is violated.
We give no formal treatment of blame in this paper, but our separation
into upcasts and downcasts naturally supports a definition of blame
analogous to theirs.
In their paper, each component $c$ is idempotent and satisfies $c
\dynr \id$.  
Their work is fundamentally untyped so a direct comparison is
difficult.
\iflong
Their pairs of projections are not coreflections between the untyped
domain and itself and it doesn't make sense to ask whether our upcasts
and downcasts are error projections because they are not
endomorphisms.
We can say that on the one hand any coreflection with components $u, d
: A \tl \dyn$ produces an error projection $u \circ d$ on $\dyn$, but
then we are left with a single projection rather than two.
We might be able to make a more direct comparison using a
\emph{semantic} type system over an untyped language, in the style of
\cite{constable+86nuprl-book}.
\fi

Recent work on interoperability in a (non-gradual) dependently typed
language \cite{dagand:hal-01629909} defines several variations of
Galois connections to serve as models of casts with different
properties.
This work validates their comments that ordinary monotone Galois
connections serve as a model of the upcasts and downcasts associated
to type dynamism.

\iflong
\subsection{Frameworks for Gradual Typing}
\fi

There are two recent proposals for a more general theory of gradual
typing: Abstracting Gradual Typing (AGT) \cite{AGT} and the
Gradualizer \cite{gradualizer17}.
Broadly, their systems and ours are similar in that type dynamism and
graduality are central and a gradually typed language is constructed
from a statically typed language.
Gradual type theory is quite different in that it is based on an
axiomatic semantics, whereas both of theirs are based on operational
semantics.
As such our notion of gradual type soundness is stronger than theirs:
we assert program equivalences whereas their soundness theorem is
related to the syntactic type soundness theorem of the static
language.
Their systems also develop a \emph{surface syntax} for gradually typed
languages (including implicit casts and gradual type checking), whereas
our logic here only applies to the \emph{runtime semantics} of the
language.
\iflong
In particular, their languages have \emph{implicit} casts which are
elaborated into an explicit cast calculus that is more similar to our
type theory.
Their approaches also consider the problem of how a gradual type
\emph{checker} should balance the demands of disallowing terms that
will produce type errors with the requirement that the language still
have a subset that supports a dynamically typed programming style.
\fi

The AGT framework also allows for a variation on gradual typing where
only some types are ``gradual'' in the sense that they are less
dynamic than the dynamic type. For instance, by removing the rule $A
\to B \dynr \dyn$ we get a dynamic type that only embeds first-order
types, and so rules out costly higher-order casts.
We can accommodate this in our axiomatics by simply limiting the $A
\dynr \dyn$ rule to only apply to certain types $A$, the definition
of the casts can remain the same.
In fact our first-order model or pointed preorders in
Section~\ref{sec:models:first-order} would be a model of such a system
since the category that interprets terms is cartesian closed.
Finally, AGT is based on abstract interpretation and uses a Galois
connection between gradual types and sets of static types that is
actually a coreflection itself, but we do not see a precise
relationship to our use of coreflections, and so this may just reflect
the ubiquity of this mathematical concept.

\iflong
\subsection{Cast Factorization}

The factorization of an arbitrary cast $A \Rightarrow B$ into an
upcast to $\dyn$ followed by a downcast is superficially similar to
the work of \cite{siek-wadler10}, which collapse a
sequence of casts starting at $A$ and ending at $B$ into a
downcast to $A \sqcap B$ followed by an upcast to $B$.
Note that their factorization is in fact opposite: ours is an upcast
followed by a downcast.
The factorization we present is trivial and was originally presented
in \cite{henglein94:dynamic-typing}, whereas theirs involves some
actual computation of a type and is similar to \emph{image
  factorization}.
Furthermore, it was shown in \cite{AGT} that the correctness of
factorization through $A \sqcap B$ is not always possible and is
highly dependent on the available language of gradual types, whereas
our factorization solely depends on the presence of a dynamic type,
which could even be weakened to the two types having a common
$\dynr$-supertype.

\fi

Relative to this related work, we believe the axiomatic specification of
casts via a universal property relative to dynamism is a new idea in
gradual typing, as is our categorical semantics and the
presentation of the contract interpretation as a model construction.

\iflong
\subsection{Future Work}
\fi

The clearest challenges for future work are the axiomatization of
gradual typing with more advanced typing features.
For instance, the combination of gradual typing and parametric
polymorphism has proven quite complex \cite{ahmed08:paramseal,
  ahmed17, igarashipoly17,torolabradatanter19}.
If we could show that the combination of graduality with parametricity
has a unique implementation, as we have shown here for simple typing,
it would provide a strong semantic justification for a design.

Additionally, the combination of dependent typing with gradual typing
is worth exploring, especially because while dependent contract
checking been used for some time and was explicitly inspired in part
by dependent typing, no semantic connection has been established
between dependent contracts and category-theoretic models of dependent
typing.
The main difficulty will be the combination of dependent typing with
effects, but there has been much recent work in this area
\cite{ahmanghaniplotkin16,Vakarthesis,pedrottabareau17}.

Another interesting application would be to apply our semantics to
other forms of gradualization, such as effect typing
\cite{gradeffects2014}, security typing \cite{toro18:sec} and refinement
typing \cite{lehmann17}.

We conjecture that much of our semantics
will hold over, but with the dynamism and casts being in a different
(preorder) category. For instance, effect types can be interpreted as
monads and a cast might be interpreted as an embedding-projection pair
of morphisms of monads.

\bibliography{max}

\end{document}

%% file: theorems-long.tex
\section{Theorems and Constructions in Gradual Type Theory}
\label{sec:theorems}

In this section, we discuss some of the consequences of the axioms of
gradual type theory in the form of \emph{theorems} derivable in
gradual type theory.
These theorems come in the form of term orderings and equivalences.
If we accept that gradual type theory is a reasonable axiomatization
of a gradually typed language satisfying $\beta\eta$ equivalence and
graduality, then the equivalences we derive here imply program
equivalences in any such language.
We divide the theorems we show into two groups. First, we present the
``reductions'', i.e., equivalences that correspond to operational
reductions in gradually typed languages.
Second, we present theorems that are not operational reductions, but
more abstract properties of casts, that help us relate to the
denotational semantics and other presentations of graduality.

\subsection{Cast Reduction Theorems}

We now present several theorems that are typically the part of the
operational semantics of a gradually typed language.
Since we are able to \emph{derive} these as theorems, this shows that
they are essential components of a language satisfying $\beta\eta$ and
graduality.
For instance, if any of these program equivalences is violated, then a
language must violate $\beta,\eta$ or graduality.

First we show that the upcast and downcast from a type to itself are
the identity function.
\begin{theorem}[Identity Casts]
  \label{case:id}  
  ${\upcast{A}{A}{t} \equidyn t}$ and ${\dncast{A}{A}{t} \equidyn t}$.
\end{theorem}
\proof
The intuition is simple: given $t : A$, $t$ \emph{itself} is the least
dynamic element of $A$ that is at least as dynamic as $t$.
For a formal proof, we show that $x : A$ and $\upcast{A}{A}{x}$ are
equi-dynamic and each direction is an instance of \textsc{UL} or
\textsc{UR}:
\begin{mathpar}
  \inferrule*[right=UR]
      {~}
      {x : A \vdash x \dynr \upcast{A}{A}{x} : A}\and

  \inferrule*[right=UL]
      {~}
      {x : A \preciser x : A \vdash \upcast{A}{A}{x} \dynr x : A \preciser A}
\end{mathpar}
The downcast has a perfectly dual proof.
\qed

Since this is our first example of using the cast term dynamism rules,
it is instructive to note that, given $A \preciser A'$ so that
$\upcast{A}{A'}{}$ is well-defined, we \emph{cannot} show that
$\upcast{A}{A'}{x} \dynr x$ analogously to the second derivation
\[
  \inferrule*[right=not an instance of UL]
      {~}
      {x : A \dynr y : A \vdash \upcast{A}{A'}{x} \dynr y : A' \preciser A}
\]
because the conclusion violates the presupposition of the judgment,
which would require $A' \dynr A$, and is moreover not an instance of
$\textsc{UL}$, which would require $y$ to have
type $A'$, not type $A$.  That is, the existence of appropriate type
dynamism relations is crucial to these rules, so it is important to be
careful about the types involved.

Next, we show that if $A \dynr A' \dynr A''$, then the upcast
from $A$ to $A''$ factors through $A'$, and dually for the downcast
from $A''$ to $A$.
This justifies the operational rule familiar in gradual typing that
separates the function contract into the ``higher-order'' part that
proxies the original function and the ``first-order'' tag checking:
\[ \upcast{A\to B}{\dyn}{t} \mapsto \upcast{\dyn\to\dyn}{\dyn}{\upcast{A\to B}{\dyn \to \dyn}{t}} \]
More generally, it implies that casts from $A$ to $A'$ where $A \dynr A'$ commute
over the dynamic type, e.g.
$\upcast{A'}{?}{\upcast{A}{A'}{x}} \equidyn \upcast{A}{?}{x}$---intuitively,
if casts only perform checks, and do not change values, then a value's
representation in the dynamic type should not depend on how it got
there.  This can also justify some \emph{optimizations} of gradual programs,
collapsing multiple casts into one.
This property, combined with the identity property, also says that
upcasts and downcasts form respective \emph{subcategories} of
arbitrary terms (the composition of two upcasts (downcasts) is an
upcast (downcast) and identity terms are also upcasts and downcasts),
and that the upcasts and downcasts each determine functors from the
category of types and type dynamism relations to the category of types
and terms.

\begin{theorem}[Casts (De-)Compose]
  \label{case:comp}
  If $A \dynr A' \dynr A''$, then
  $\upcast{A}{A''}{t} \equidyn \upcast{A'}{A''}{\upcast{A}{A'}{t}}$
  and dually,
  $\dncast{A}{A''}{t}\equidyn{\dncast{A}{A'}{\dncast{A'}{A''}{t}}}$.
\end{theorem}
\proof
The proofs are dual, so we only show the argument for upcasts.
We want to show $\upcast{A}{A''}{x}\equiprecise
\upcast{A'}{A''}{\upcast{A}{A'}{x}}$.
On the one hand, to show something of type $A''$ is more dynamic than
$\upcast{A}{A''}{x}$, we just have to show that it is more dynamic
than $x$, which is true of $\upcast{A'}{A''}{\upcast{A}{A'}{x}}$.
The other direction is similar, first we peel off
$\upcast{A'}{A''}{\cdot}$ and then $\upcast{A}{A'}{\cdot}$.
More formally, assuming $A \dynr A' \dynr A''$, the following are valid
derivations:
\begin{mathpar}
  \inferrule*[right=UL(S)]
  {\inferrule*[right=UR(S)]
    {\inferrule*[right=UR]
      {~}
      {x : A \vdash x \dynr \upcast{A}{A'}{x} : A \dynr A'}}
    {x : A\vdash x \dynr \upcast{A'}{A''}{\upcast{A}{A'}{x}} : A \dynr A''}}
  {x : A \vdash \upcast{A}{A''}{x} \dynr \upcast{A'}{A''}{\upcast{A}{A'}{x}} : A''}\and

  \inferrule*[right=UL(S)]
  {\inferrule*[right=UL(S)]
    {\inferrule*[right=UR]
      {~}
      {x : A \vdash x \dynr \upcast{A}{A''}{x} : A \dynr A''}}
    {x : A \vdash \upcast{A}{A'}{x} \dynr \upcast{A}{A''}{x} : A' \dynr A''}}
  {x : A \vdash \upcast{A'}{A''}{\upcast{A}{A'}{x}}\dynr \upcast{A}{A''}{x} : A''}
\end{mathpar}
\qed

Next, we show the most important theorems, which state that casts between
function types must be implemented by the functorial action of the
function type on casts (and the analogous case for products).
This reproduces the standard ``wrapping'' implementation of function
and product casts \cite{findler-felleisen02, findler-blume06}.
Each proof is modular (depends on no more type or term constructors
other than those related to function and product types respectively).

\begin{theorem}[Function and Product Cast Reductions]
  Whenever $A \dynr A'$ and $B \dynr B'$, the following are all
  satisfied.
  \begin{enumerate}
  \item $\upcast{A \to B}{A'\to B'}{t} \equidyn \lambda x:A'. \upcast{B}{B'}{(t (\dncast{A}{A'}{x}))}$
  \item $\dncast{A \to B}{A'\to B'}{t} \equidyn \lambda x:A. \dncast{B}{B'}{(t (\upcast{A}{A'}{x}))}$.
  \item $\upcast{A_0 \times A_1}{A_0' \times A_1'}{t} \equidyn (\upcast{A_0}{A_0'}{\pi_0 t}, \upcast{A_1}{A_1'}{\pi_1 t})$
  \item $\dncast{A_0 \times A_1}{A_0' \times A_1'}{t} \equidyn (\dncast{A_0}{A_0'}{\pi_0 t}, \dncast{A_1}{A_1'}{\pi_1 t})$.
  \end{enumerate}
\end{theorem}
\proof

To make the function proof easier to understand, we first derive a
higher-level ``extensionality principle'', which the reader may find
to be more intuitive than the $\eta$ principle: a function is
less dynamic than another if applying it to a less dynamic input
yields a less dynamic result:
\begin{mathpar}
    \inferrule*[right=fun-ext]
    {\Phi, x \dynr x' : A \dynr A' \vdash t x \dynr t' x' : B \dynr B'}
    {\Phi \vdash t \dynr t' : A \to B \dynr A' \to B'}
\end{mathpar}
It follows from the $\eta$ principles and the
congruence rules for $\lambda$:
  \begin{mathpar}
    \inferrule*[right=tmprec-trans]
    {t \dynr \lambda x. t x \and
      \lambda x'. t' x' \dynr t'\and
      \inferrule*
      {\Phi, x \dynr x' : A \dynr A' \vdash t x \dynr t' x' : B \dynr B'}
      {\Phi \vdash \lambda x. t x \dynr \lambda x'. t' x' : A \to B \dynr A' \to B'}
    }
    {\Phi \vdash t \dynr t' : A \to B \dynr A' \to B'}\\
  \end{mathpar}

For the function contract, we need to show
\[
\upcast{A\to B}{A'\to B'}{f} \equiprecise \funupcast{A}{B}{A'}{B'}{f}{x'}.
\]
First to show $\dynr$, it is sufficient to show that the right
hand side is more dynamic than $f$ itself.
Next we invoke the extensionality principle (\textsc{fun-ext}) and
$\beta$ and then we have to show that $x \dynr x' : A \dynr A'
\vdash f x \dynr \upcast{B}{B'}{(f (\dncast{A}{A'}{x'}))}$.
This follows from congruence of application and the rules of casts.
As a derivation tree:
\begin{mathpar}
  \inferrule*
  {
    \inferrule*
    {
      \inferrule*
      {
        \inferrule*
        {
          \inferrule*
          {f, x \dynr x' \vdash f \dynr f \and
            \inferrule*
            {f, x \dynr x' \vdash x \dynr x'}
            {f, x \dynr x' \vdash x \dynr \dncast{A}{A'}x'}
          }
          {f, x \dynr x' \vdash f x \dynr f (\dncast{A}{A'}x')}
        }
        {f, x \dynr x' \vdash f x \dynr \upcast{B}{B'}(f (\dncast{A}{A'}x'))}
      }
      {f, x \dynr x' : A \dynr A' \vdash f x \dynr (\funupcast{A}{B}{A'}{B'}{f}{x'})x'}
    }
    {f \vdash f \dynr \funupcast{A}{B}{A'}{B'}{f}{x'}}
  }
  {f : A\to B \vdash \upcast{A\to B}{A'\to B'}{f} \dynr \funupcast{A}{B}{A'}{B'}{f}{x'}}
\end{mathpar}

For the opposite direction, we invoke the extensionality principle and
$\beta$ reduce, then needing to show $\upcast{B}{B'}(f
(\dncast{A}{A'}x')) \dynr (\upcast{A \to B}{A'\to B'}f)x'$. We can
remove the upcast from the left and then use the congruence rules.
As a derivation tree:
\begin{mathpar}
  \inferrule
  {
    \inferrule
    {
      \inferrule
      {
        \inferrule
        {
          \inferrule
          {f, x' \vdash f \dynr f}
          {f, x' \vdash f \dynr \upcast{A\to B}{A'\to B'}{f}}
          \and
          \inferrule
          {f , x' \vdash x' \dynr x'}
          {f , x' \vdash\dncast{A}{A'}x' \dynr x'}
        }
        {f, x' \vdash f (\dncast{A}{A'}x') \dynr (\upcast{A\to B}{A'\to B'}{f})x'}
      }
      {f, x' \vdash\upcast{B}{B'}(f (\dncast{A}{A'}x')) \dynr (\upcast{A\to B}{A'\to B'}{f})x'}
    }
    {f, x' : A' \vdash (\funupcast{A}{B}{A'}{B'}{f}{x'})x'\dynr (\upcast{A\to B}{A'\to B'}{f})x'}
  }
  {f : A\to B \vdash \funupcast{A}{B}{A'}{B'}{f}{x'}\dynr \upcast{A\to B}{A'\to B'}{f}}
\end{mathpar}

The downside of using the extensionality principle is that we need to
use $\beta$ reduction, when in actuality we only need to use $\eta$ in
the proof. In figure~\ref{fig:negative-upcast}, we present ``direct''
proofs for function and product upcasts that use only $\eta$
equivalence, and not the extensionality principle or $\beta$. The
proofs for downcasts are exactly dual.
\begin{figure}
  \begin{small}
\begin{mathpar}
    \inferrule
          {\inferrule*
            {f \dynr \lambda x. f x \and
              \inferrule*
                  {\inferrule*
                    {\inferrule*
                      {f, x \dynr x' \vdash f \dynr f : A \to B \dynr A \to B \and
                        \inferrule*
                            {f, x \dynr x' \vdash x \dynr x' : A \dynr A'}
                            {f, x \dynr x' \vdash x \dynr \dncast{A}{A'}{x'} : A \dynr A}
                      }
                      {f, x \dynr x' \vdash f x \dynr f (\dncast{A}{A'}{x'}) : B \dynr B}}
                    {f, x \dynr x' \vdash f x \dynr \upcast{B}{B'}{(f (\dncast{A}{A'}{x'}))} : B \dynr B'}}
                  {\lambda x. f x \dynr \funupcast{A}{B}{A'}{B'}{f}{x'}}
            }
            {f \vdash f \dynr \funupcast{A}{B}{A'}{B'}{f}{x'} : A \to B \dynr A' \to B'}}
          {f : A \to B \vdash \upcast{A\to B}{A'\to B'}f \dynr \funupcast{A}{B}{A'}{B'}{f}{x'} : A' \to B'}    \\\\
  \end{mathpar}    
  \end{small}  
  \begin{mathpar}
          \inferrule*[right=$\mathcal{D}$]
                {\inferrule*
                  {\inferrule*
                    {f \vdash f \dynr (\upcast{A\to B}{A'\to B'}f) \and
                      x' \vdash \dncast{A}{A'}{x'} \dynr x'
                    }
                    {f, x' \vdash f (\dncast{A}{A'}{x'}) \dynr (\upcast{A\to B}{A'\to B'}f) x'}}
                  {f, x' : A' \vdash \upcast{B}{B'}{(f (\dncast{A}{A'}{x'}))} \dynr (\upcast{A\to B}{A'\to B'}f) x'}}
                {f \vdash \funupcast{A}{B}{A'}{B'}{f}{x'}\dynr \lambda x'. (\upcast{A\to B}{A'\to B'}f) x'}\\\\
          \inferrule*
              { \mathcal{D}
                \and
                \lambda x'. t x' \dynr t \text{ with } t = \upcast{A\to B}{A'\to B'}{f}}
              {f : A \to B \vdash \funupcast{A}{B}{A'}{B'}{f}{x'}\dynr \upcast{A\to B}{A'\to B'}f : A' \to B'}

              \\\\

              \inferrule*
                  {\inferrule*
                    {p \dynr (\pi_0 p, \pi_1 p) \and
                      \inferrule*
                          {\forall i \in {0,1}.
                            \inferrule*{\pi_i p \dynr \pi_i p}
                                       {\pi_i p \dynr \upcast{A_i}{A_i'}\pi_i p}}
                          {p \vdash (\pi_0 p, \pi_1 p) \dynr \produpcast{A_0}{A_1}{A_0'}{A_1'}{p}}
                    }
                    {p \vdash p \dynr \produpcast{A_0}{A_1}{A_0'}{A_1'}{p}}}
                  {p : A_0 \times A_1 \vdash \upcast{A_0\times A_1}{A_0' \times A_1'}{p} \dynr \produpcast{A_0}{A_1}{A_0'}{A_1'}{p} : A_0' \times A_1'}

                  \\\\
  \end{mathpar}
  \begin{small}
    \begin{mathpar}    
      \inferrule*[right=$\mathcal{E}$]
                 {\forall i \in {0,1}.
                   \inferrule*
                       {\inferrule*
                         {\inferrule*
                           {p \dynr p}
                           {p \dynr \upcast{A_0\times A_1}{A_0' \times A_1'}{p}}}
                         {\pi_i p \dynr \pi_i\upcast{A_0\times A_1}{A_0' \times A_1'}{p}}}
                       {\upcast{A_i}{A_i'}{\pi_i p} \dynr \pi_i\upcast{A_0\times A_1}{A_0' \times A_1'}{p}}
                 }
                 {\produpcast{A_0}{A_1}{A_0'}{A_1'}{p} \dynr (\pi_0\upcast{A_0\times A_1}{A_0' \times A_1'}{p}, \pi_1\upcast{A_0\times A_1}{A_0' \times A_1'}{p})}\\\\
    \end{mathpar}
  \end{small}
  \begin{mathpar}
                  \inferrule*
                      { \mathcal{E}
                        \and
                        (\pi_0t, \pi_1t) \dynr t \text{ with }  t = \upcast{A_0\times A_1}{A_0' \times A_1'}{p}
                      }
                      {p : A_0 \times A_1 \vdash \produpcast{A_0}{A_1}{A_0'}{A_1'}{p} \dynr \upcast{A_0\times A_1}{A_0' \times A_1'}{p} : A_0' \times A_1'}
    \end{mathpar}
  \caption{Function, Product Upcast Reduction Derivations}
  \label{fig:negative-upcast}
\end{figure}
\qed

We note that in the proofs above, each direction of the equivalence
depends only on \emph{one} direction of $\eta$ equivalence, and so in
gradual languages where $\eta$ only holds as an ordering, we still get
an ordering relationship with the wrapping semantics.

Finally, we can show that upcasts must preserve errors, and if the
retract axiom is assumed, downcasts must as well.
\begin{theorem}[Casts are strict]
  If $A \dynr A'$, then
  \begin{enumerate}
  \item $\upcast{A}{A'}{\err_A} \equidyn \err_{A'}$
  \item Assuming the retract axiom, $\dncast{A}{A'}{\err_A} \equidyn \err_{A'}$
  \end{enumerate}
\end{theorem}
\begin{proof}
The upcast preserves $\err$ because it is a left/upper adjoint and
therefore preserves colimits/joins, and $\err$ is the empty join.
More concretely, $\upcast{A}{A'}{\err_A} \dynr \err_{A'}$ because
$\err_{A'}$ is more dynamic than $\err_A$ and is the least dynamic
term of type $A'$ so is in particular less dynamic than anything more
dynamic than $\err_A$.  As derivations:

\begin{mathpar}
  \axiom{Err-Bot}{\cdot \vdash \err_{A'} \dynr \upcast{A}{A'}{\err_{A}} : A'}

  \inferrule*[right=UL(S)]
  {\axiom{Err-Bot'}
    {\cdot \vdash \err_{A} \dynr \err_{A'} : A \preciser A'}
  }
  {\cdot \vdash \upcast{A}{A'}{\err_{A}} \dynr \err_{A'} : A'}
\end{mathpar}

The proof that the downcast preserves $\err$ is less modular as it
depends on the presence of the upcast and the retract axiom.
The proof is simple though: to show $\dncast{A}{A'}{\err_{A'}} \equiprecise
\err_{A} : A$, we have $\err_{A'} \equiprecise \upcast{A}{A'}{\err_{A}}$
by above, and
so we can apply the downcast to both sides
to get
$\dncast{A}{A'}{\err_{A'}} \equiprecise \dncast{A}{A'}{\upcast{A}{A'}{\err_{A}}}$,
and the right-hand side is equivalent
to $\err_{A}$ by the retract axiom.
\end{proof}

\subsection{Properties of Casts}

Next we present a few properties of casts that are not operational
reductions, but have other semantic significance.

First, we make precise the idea that the specification we have given
for casts characterizes them uniquely up to order-equivalence
($\equidyn$).
In category theoretic terms, we show that the rules for
upcasts/downcasts constitute a \emph{universal property}.
First, to prove that casts are unique, suppose that there was a second
version of the upcast $\upcastpr{A}{A'}{t}$ with analogous sequent-style rules
\textsc{UL(S)'} and \textsc{UR(S)'}.

\begin{theorem}[Upcasts are Unique]
  \label{thm:syntactic-casts-are-uniq}
  If $A \dynr A'$, then $\upcast{A}{A'}x \equidyn \upcastpr{A}{A'}x$.
\end{theorem}
\begin{proof}
We show that the second upcast is equivalent to the original in
analogy with the way we show function/product types are unique: use
the ``elimination'' rule of one and then the ``introduction'' rule of
the other.
\begin{mathpar}
  \inferrule*[right=UR(S)]
  {\inferrule*[right=UR(S)']
   {x : A \vdash x \preciser x : A}
   {x : A \vdash x \preciser \upcastpr{A}{A'}{x} : A \dynr A'}}
  {x : A \vdash \upcast{A}{A'}{x} \preciser \upcastpr{A}{A'}{x} : A'}

  \inferrule*[right=UR(S)']
  {\inferrule*[right=UR(S)]
   {x : A \vdash x \preciser x : A}
   {x : A \vdash x \preciser \upcast{A}{A'}{x} : A \dynr A'}}
  {x : A \vdash \upcastpr{A}{A'}{x} \preciser \upcast{A}{A'}{x} : A'}
\end{mathpar}  
\end{proof}
\noindent A dual proof gives uniqueness for downcasts as well.  

Next, we show that the upcast and downcast for a given $A \dynr A'$
always form a \emph{Galois Connection}.
This fact forms the basis for our models considered in sections
\ref{sec:contract-translation} and \ref{sec:models}.

\begin{theorem}[Casts are Galois Connections]
  \label{thm:galois}
  If $A \dynr A'$, then $t \dynr \dncast{A}{A'}{\upcast{A}{A'}{t}}$
  and $\upcast{A}{A'}{\dncast{A}{A'}{t}}\dynr t$
\end{theorem}
\begin{proof}
  The derivations are as follows:
  \begin{mathpar}
\inferrule*[right=UL(S)]
             {{x : A \vdash x \dynr x : A \dynr  A}}
             {\inferrule*[right=DR(S)]
                         {x : A \vdash x \dynr {\upcast{A}{A'}{x}} : A \dynr  A'}
                         {x : A \vdash x \dynr \dncast{A}{A'}{\upcast{A}{A'}{x}} : A}}
\and
\inferrule*[right=DL(S)]
             {x' : A' \vdash x' \dynr x' : A' \dynr  A'}
             {\inferrule*[right=UL(S)]
                         {x' : A' \vdash {\dncast{A}{A'}{x'}} \dynr x'  : A \dynr  A'}
                         {x' : A' \vdash \upcast{A}{A'}{\dncast{A}{A'}{x'}} \dynr x' : A'}}
\end{mathpar}
\end{proof}
In programming practice, we expect the stronger property that round
trip from $A$ to $A'$ and back to be in fact an identity and this is
implied by the addition of the \emph{retract axiom}.

Recall that the gradual guarantee~\cite{refined} says that making
casts less dynamic results in semantically less dynamic terms,
but does not otherwise change the behavior of programs.
To see that a model of gradual type theory satisfies the gradual
guarantee, the key syntactic fact is that making \emph{casts} less dynamic
results in a term dynamism relationship.
We state this as the following theorem in GTT:
\begin{theorem}[Cast Congruence]
  When $A \dynr A', B \dynr B', A \dynr B, A' \dynr B'$, we can derive
  \[ x \dynr y : A \dynr B \vdash \upcast{A}{A'}{x} \dynr \upcast{B}{B'}{y} : A' \dynr B'\]
  and
  \[ x' \dynr y' : A' \dynr B' \vdash \dncast{A}{A'}{x'} \dynr \dncast{B}{B'}{y'} : A \dynr B .\]
\end{theorem}
\begin{proof}
The proof of the
first is
\begin{mathpar}
\inferrule*[right=UR(S)]
           {
             {x \dynr y : A \dynr B \vdash x \dynr y : A \dynr B}
           }
{\inferrule*[right=UL(S)]
           {x \dynr y : A \dynr B \vdash x \dynr \upcast{B}{B'}{y} : A \dynr B'}
           {x \dynr y : A \dynr B \vdash \upcast{A}{A'}{x} \dynr \upcast{B}{B'}{y} : A' \dynr B'}}
\end{mathpar}
and the second is dual.
\end{proof}

All other term constructors are congruences by primitive rules, so
$\dynr$ is a congruence.

Finally, we discuss what it means for two types to be equivalent in gradual type theory.
Because types $A$ and $B$ in gradual type theory can be related both
by type dynamism $A \dynr B$ and by functions $A \to B$, there are two
reasonable notions of equivalence of types.
First, if they are order-equivalent:
\begin{definition}[Equi-dynamic]
  We say types $A$ and $B$ are \emph{equi-dynamic} if $A \dynr B$ and $B \dynr A$
\end{definition}
And second, if they are isomorphic:
\begin{definition}[Isomorphism]
  We say types $A$ and $B$ are \emph{isomorphic} if there exist terms
  $x : A \vdash t : B$ and $y : B \vdash u : A$ such that
  \[ x:A \vdash u[t/y] \equidyn x : A\]
  and
  \[ y:B \vdash t[u/x] \equidyn y : B\]
\end{definition}
Equi-dynamism of types turns out to be strictly stronger.
First, equi-dynamism of types implies isomorphism, with the casts
between them forming the isomorphism:
\begin{theorem}[Equi-dynamic Types are Isomorphic]
  If $A \equidyn B$, then
  \begin{enumerate}
\item $x. \upcast{A}{B}{x}$ and $y. \upcast{B}{A}{y}$ form an
    isomorphism of types.
\item $y. \dncast{A}{B}{x}$ and $x. \dncast{B}{A}{y}$ form an
    isomorphism of types.
\item $x. \upcast{A}{B}{x}$ and $y. \dncast{A}{B}{y}$ form an
    isomorphism of types.  
\item $y:B \vdash \upcast{B}{A}{x} \equidyn \dncast{A}{B}{x} : A$.
\end{enumerate}
\end{theorem}
\begin{proof}
  \begin{enumerate}
  \item By the following derivations:

    \begin{mathpar}
    \inferrule*[right=UL(S)]{
      \inferrule*[right=UL(S)]{x : A \vdash x \dynr x : A}
                {x : A \vdash \upcast{A}{B}{x} \dynr x : B \dynr A}
    }{x : A \vdash \upcast{B}{A}{\upcast{A}{B}{x}} \dynr x : A}

    \inferrule*[right=UR(S)]
    {\inferrule*[right=UR(S)]{x : A \vdash x \dynr x : A}
      {x : A \vdash x \dynr \upcast{A}{B}{x} : A \dynr B}}
    {x : A \vdash x \dynr \upcast{B}{A}{\upcast{A}{B}{x}} : A}
\end{mathpar}
  \item This is dual to the previous case.
  \item By the following derivation:
    \begin{mathpar}
    \inferrule*[right=DL(S)]
    {\inferrule*[right=UL(S)]
     {x : A \vdash x \dynr x : A}
     {x : A \vdash \upcast{A}{B}{x} \dynr x : B \dynr A}
    }
    {x : A \vdash \dncast{A}{B}{\upcast{A}{B}{x}} \dynr x : A}

    \inferrule*[right=DR(S)]
    {\inferrule*[right=UR(S)]{x:A \vdash x \dynr x : A}{x : A \vdash x \dynr \upcast{A}{B}{x} : A \dynr B}}
    {x : A \vdash x \dynr \dncast{A}{B}{\upcast{A}{B}{x}} : A}
\end{mathpar}
  \item This follows from uniqueness of inverses (which is true by the
    usual argument) and the previous two.
    \qedhere
  \end{enumerate}
\end{proof}

The converse, that isomorphic types are equi-dynamic, does not hold
by design, because it does not match gradual typing practice.
Gradually typed languages typically have \emph{disjointness} of
connectives as operational reductions; for example, disjointness of
products and functions can be expressed by an axiom
$\dncast{(C \times D)}{\dyn}{\upcast{(A \to B)}{\dyn}{x}}
\dynr \err$ which says that casting a function to a product
errors.
This axiom is incompatible with isomorphic types being
equi-dynamic, because a function type \emph{can} be isomorphic to a
product type (e.g. $X \to Y \cong (X \to Y) \times 1$), and for
equi-dynamic types $A$ and $B$, a cast
$\dncast{{B}}{\dyn}{\upcast{A}{\dyn}{x}}$ cannot fail, because if
it fails, then every term of $A$, $B$ equals $\err$:
Assume
$A \equidyn B$
and 
$\dncast{B}{\dyn}{\upcast{A}{\dyn}{x}} \dynr \err$.
By composition and the adjunction property
\[
\upcast{A}{B}{x} \dynr
\dncast{B}{\dyn}{\upcast{B}{\dyn}{\upcast{A}{B}{x}}} \equidyn
\dncast{B}{\dyn}{\upcast{A}{\dyn}{x}} \dynr
\err
\]
But by above,
$\upcast{B}{A}{}$ is an isomorphism, so
\[
x:A \equidyn \dncast{A}{B}{\upcast{A}{B}{x}} \dynr \dncast{B}{A}{\err} \equidyn \err
\]
where the last step is by strictness, so every element of $A$ (and $B$, by congruence
of casts) is equal to a type error.
That is, disjointness axioms make equi-dynamism an intensional property
of the representation of a type, and therefore stronger than
isomorphism.
Nonetheless, the basic rules of gradual type theory do not imply
disjointness; in Section~\ref{sec:models}, we discuss a countermodel.



%% file: models-long.tex
\section{Concrete Models}
\label{sec:models}

Now that we have identified a general method of constructing models of
gradual type theory, we can produce some concrete models by producing
suitable 2-categories.

\subsection{Pointed Preorder Model}
\label{sec:models:first-order}

First, we present a simple first-order preorder model.
The model is first-order because it models the fragment of gradual
type theory without function types.
However, by not accommodating function types it is much more
elementary.
The 2-category for the preorder model is the category $PreOrd_{\bot}$
whose objects are preorders with a least element, which following
domain-theoretic terminology we call ``pointed'' preorders, and whose
morphisms are monotone functions (that don't necessarily preserve
$\bot$) and 2-cells are given by the obvious ordering on morphisms.
This is a cartesian locally thin 2-category with local $\bot$s (also
closed but we will not use this).
To construct a suitable dynamic type, we can start with a base set,
such as the natural numbers $\mathbb{N}$ and construct the dynamic type
by finding a solution to the equation:
\[
D \cong \mathbb{N}_{\bot} \oplus (D \times D)
\]
where $\oplus$ is the wedge sum of pointed preorders that identifies
the $\bot$s of the two sides.
Since this is a \emph{covariant} domain equation, this can be
constructed as a simple colimit, and the solution has as elements
finite binary trees whose leaves are either natural numbers or a base
element $\bot$.
The ordering on the trees $T \dynr T'$ holds when $T$ can be
produced from $T'$ by replacing some number of subtrees by $\bot$,
a simple model of the dynamism ordering.
Finally to get a model, the upcast of the coreflection $D \times D \tl
D$ simply injects to the right side of $\oplus$ and the downcast
errors on the $\mathbb{N}_{\bot}$ case and otherwise returns the pair.

\subsection{Scott's Model}
\label{sec:models-scott}

Next we present two models based on domains that are operationally
\emph{inadequate} because they identify the dynamic type error and
diverging programs.
The first is merely a new presentation of Dana Scott's classical
models of untyped lambda calculus but for a gradually typed language
\cite{scott76}.
The second is a variation on that construction where product and
function types have overlapping representation, showing that the
product and function types cannot be proven disjoint in gradual type
theory.
Both are based on the 2-category of pointed $\omega$-chain complete
partial orders, which we simply call \emph{domains} and continuous
functions.
By standard domain-theoretic techniques (see \cite{wand79ocat,
  smythplotkin, pitts96}) we can construct a suitable dynamic type by
solving the recursive domain equation:
\[ D \cong \mathbb{N}_{\bot} \oplus (D \times D) \oplus (D \to D) \]
where $\oplus$ is the wedge sum of domains that identifies their least
element.
The classical technique for solving this equation naturally produce the
required coreflections $(D\times D) \tl D$ and $(D \to D) \tl D$.

Next, to get a model in which product and function types are not
disjoint we can construct a dynamic type as a \emph{product} of our
connectives rather than a sum:
\[ D' \cong \mathbb{N}_{\bot} \times (D'\times D') \times (D' \to D') \]

This is a kind of ``coinductive'' dynamic type that can be thought of
as somewhat object-oriented: rather than an element of the dynamic
type being a tagged value, it is something that responds to a
set of messages (given by the projections) and if it
``doesn't implement'' the message it merely returns $\bot$.
Then
$\dncast{(\dyn\times\dyn)}{\dyn}{\upcast{(\dyn\to\dyn)}{\dyn}{x}} \not\equiv
\err$ because there are elements of the domain that are non-trivial
both in the $D\times D$ position and $D \to D$ position.

\subsection{Resolution: Pointed Domain Preorders}

We can combine the best aspects of the domain and pointed preorder
models into a single model of pointed, preorder domains, i.e., domains
that in addition to their intrinsic domain ordering that models a
``divergence ordering'' with diverging programs modeled by the
divergence-least element have a second, ``error ordering'' with a
least element $\err$ that models the dynamic type error.
The error ordering needs two properties related to the domain
structure.
First, the diverging element should be maximal, meaning that it is not
considered to be more erroring than anything else, which ensures that
when we take the wedge sum, we \emph{only} identify $\bot$s and don't
otherwise affect the error ordering.
Second, the error ordering should be an \emph{admissible} relation,
meaning that a limit of one chain is error-less-than the limit of
another, if at every step they are related.

\begin{definition}[Domain Preorder, Pointed Domain Preorder, Continuous Functions]
  \hfill
  \begin{enumerate}
  \item A domain preorder $X$ is a set $|X|$ with two orderings
    $\leq_{X}$ and $\dynr_{X}$ such that $(|X|,\leq_X)$ is an
    $\omega$-complete pointed partial order with $\leq_X$-least
    element $\bot_X \in |X|$ and $\dynr_X$ is a preorder such that
    \begin{enumerate}
    \item $\bot_X$ is a $\dynr_X$-maximal element: if $\bot \dynr x$
      then $x = \bot$
    \item $\dynr$ is $\leq$-admissible/closed under limits of
      $\leq_X$-$\omega$-chains: if $\{x_i\}_{i<\omega},
      \{y_i\}_{i<\omega}$ are $\leq_X$-$\omega$-chains in $X$ and $x_i
      \dynr y_i$ for every $i < \omega$, then $\bigvee \{x_i\} \dynr
      \bigvee \{y_i \}$.
    \end{enumerate}

  \item A continuous function of domain preorders is a function of the
    underlying sets that is continuous with respect to $\leq_X$ and
    monotone with respect to $\dynr_X$.
  \item A pointed domain preorder $X$ is a domain preorder with a
    $\dynr_X$-least element $\err \in |X|$. Continuous functions of
    pointed domain preorders are just continuous functions of the
    underlying domain preorders
  \end{enumerate}
\end{definition}

We model our types as pointed domain preorders because they have an
interpretation for both divergence $\bot$ and type error $\err$, and
our terms will be modeled as continuous functions.
First, we show that both categories are cartesian closed.

\begin{lemma}
\label{lem:domain-preorder-ccc}
  The categories of domain preorders and pointed domain preorders with
  continuous functions are cartesian closed.
\end{lemma}
\begin{proof}
  Since the category of pointed domain preorders is a full subcategory
  of domain preorders, we can show that it has unit, product and
  exponential if the category of domain preorders has the
  corresponding universal property and the operations happen to
  produce pointed domain preorders when given pointed domain
  preorders.
  \begin{enumerate}
  \item $1$ is the singleton $\{*\}$ with the unique ordering. The
    unique element is a $\leq$ and $\dynr$-least element.
  \item $X \times Y$ is constructed as follows. The underlying set is
    the product of underlying sets $|X\times Y| = |X| \times |Y|$.
    Each ordering is point-wise, and limits are taken
    pointwise. Because limits are pointwise, $\dynr_X \times \dynr_Y$
    is admissible. If $X, Y$ have $\err$s, then $X \times Y$ has a
    $\dynr$-least element $\err_{X\times Y} = (\err_X,\err_Y)$.

    The pairing, and projection functions are clearly continuous in
    $\leq$ and monotone in $\dynr$.
  \item $X \to Y$. The underlying set is the set of continuous
    functions of domain preorders from $X$ to $Y$.  The orderings are
    given pointwise.

    Since this is a subset of the set of continuous functions of the
    underlying domains, to show that it is closed under
    $\leq$-$\omega$-chains, we just need to show that the limit of a
    chain of monotone functions is monotone.  The limit of the chain
    $\{f_i\}$ is just $\lambda x. \bigvee \{f_i(x)\}$.  Given $x \dynr_X
    y$, we need to show $\bigvee\{f_i(x)\} \dynr_Y \bigvee\{f_i(y)\}$.
    Since $\dynr_Y$ is admissible, it is sufficient to show that
    $f_i(x) \dynr f_i(y)$ for each $i$, which follows by monotonicity
    of $f_i$.

    Next, the pointwise ordering is clearly admissible because limits
    are taken pointwise.

    Next, if $f : X \times Y \to Z$ is a continuous function of
    preorder domains, $\lambda_y f : X \to (Y \to Z)$ is clearly
    continuous. And the application function $app : (X\to Y) \times X
    \to Y$ is also continuous.

    If $Y$ has a $\err$, then $X \to Y$ has a least element
    $\err_{X\to Y} = \lambda x. \err_Y$.
    \qedhere
  \end{enumerate}
\end{proof}

Next, we will solve the recursive domain equations using the classic
construction using ep pairs.
Our category of domain preorders is an $O$-category in the sense of
\cite{wand79ocat, smythplotkin}, but does not have all $\omega^{op}$
limits due to the restriction that the $\bot$ element be
$\dynr$-maximal.
However, if the $\omega^{op}$ diagram is made of strict
($\bot$-preserving) functions, the limit does exist, and this is all
the central theorem of \cite{smythplotkin} actually requires because
projections are always strict.

\begin{lemma}
  The category of domain preorders and continuous functions is an
  $O$-category with the point-wise $\leq$ ordering.
\end{lemma}
\begin{proof}
  We need to show two properties
  \begin{enumerate}
  \item First, every hom set is an $\omega$-cppo as shown in
    lemma \ref{lem:domain-preorder-ccc}.
  \item Next we need to show composition is $\leq$-continuous. Since
    composition can be defined as $\lambda f. \lambda g. \lambda x.
    g(f(x))$ it is sufficient to show that $\lambda : (X\times Y \to Z)
    \to (X \to (Y \to X))$ is continuous (application is continuous
    because the category is cartesian closed).
    Given a chain of $f_i : X \times Y \to Z$, on the one hand
    \[ (\bigvee_i \{\lambda f_i\})(x)(y) = \bigvee_i\{f_i(x,y)\} \]
    and on the other hand
    \[ (\lambda \bigvee_i \{f_i\})(x)(y) = (\bigvee_i\{f_i\})(x,y) =
      \bigvee_i\{f_i(x,y)\}
      \tag*{\qedhere}
    \]
  \end{enumerate}
\end{proof}

\begin{lemma}
  The category of domain preorders has all $\omega^{op}$ limits where
  the diagram is made of $\bot$-preserving maps.
\end{lemma}
\begin{proof}
  Given an $\omega^{op}$ diagram of domain preorders $f_i : D_{i+1}
  \to D_i$ where $i \in \omega$, the proposed limit $D_{\omega}$ has
  as underlying set
  \[ |D_{\omega}| = \{ d : \Pi_{i\in\omega}D_i \alt \forall i \in \omega.~ f_i(d_{i+1}) = d_i \}\]
  The orderings $\leq,\dynr$ are both given point-wise, and
  admissibility follows from that.
  $\lambda i. \bot_{D_i} \in D_{\omega}$ because $f_i(\bot_{i+1}) =
  \bot_i$ by strictness. Since the order is pointwise, it is a least
  element and maximal.
  Finally, because the ordering is pointwise, a function into
  $D_{\omega}$ is continuous if and only if its composition with each
  projection is continuous, so $D_{\omega}$ is the limit of the given diagram.
\end{proof}

\begin{lemma}
  The product $\times$, exponential $\to$, wedge sum $\oplus$ and
  adjoining an error $\cdot_{\err}$ are all locally continuous
  mixed-variance functors on the category of domain preorders.
\end{lemma}
\begin{proof}
  \begin{enumerate}
  \item $\times$: $f\times g = \lambda (x,y). (fx,gy)$, this is
    obviously functorial and is locally continuous because the orderings
    are point-wise.
  \item $\to$: $f \to g = \lambda h. g \circ h \circ f$, obviously
    functorial and locally continuous because the ordering is
    pointwise.
  \item $X_{\err}$ The underlying set is $\{0\}\times |X| \uplus
    \{1\}\times \{\err\}$ and the orderings are defined as
    \[ (i,z) \leq (i',z') = i = i' \wedge z \leq z' \]
    and
    \[ (i,z) \dynr (i',z') = i = 1 \vee (i=i'=0 \wedge z \dynr_X z') \]
    Which is clearly a domain preorder.
    The functorial action is defined by
    \begin{align*}
      f_\err(1,\err) = \err\\
      f_\err(0,x) = (0,f(x))
    \end{align*}
    which is clearly functorial and continuous.
  \item $\oplus$. $|X\oplus Y| = (\{0\}\times X \uplus \{1\}\times
    Y)/((0,\bot_X)=(1,\bot_Y))$ with the $\leq$-ordering defined as

    \[
    (i,z) \leq_{X\oplus Y} (i',z') = (z = \bot_X)\vee z = \bot_Y \vee (i=i' \wedge z \leq z')
    \]

    And error ordering similarly defined as
    \[
    (i,z) \dynr_{X\oplus Y} (i',z') = (z,z' \in \{\bot_X,\bot_Y\}) \vee (i = i' \wedge z \dynr z')
    \]
    These define a domain and preorder structure respectively, and
    $\dynr$ is admissible. Since $\bot_X,\bot_Y$ are maximal,
    $\bot_{X\oplus Y}$ is also maximal.
    \qedhere
  \end{enumerate}
\end{proof}

We can then construct a suitable dynamic type using the construction of \cite{smythplotkin}.
\begin{theorem}
  There exists a domain preorder with an isomorphism
  \[i : D \cong (\mathbb{N}_{\bot})_\err \oplus (D_\err \times D_\err) \oplus (D_\err \to D_\err) \]
\end{theorem}

Then we interpret the dynamic type as $D_\err$.  The coreflection
$(e_1,p_1) : 1 \tl D_\err$ is the unique $\err$-coreflection: $e_1(*) =
\err$ and $p_1(x) = *$.
The coreflection $(e_\times,p_\times) : D_\err \times D_\err \tl D_\err$
is defined as
\begin{align*}
  e_\times(x) &= i((1,x))\\
  p_\times(\err) &= \err\\
  p_\times(\bot) &= \bot\\
  p_\times(0,n : (\mathbb{N}_\bot)_\err) &= \err\\
  p_\times(1,t : (D_\err \times D_\err)) &= t\\
  p_\times(2,f : (D_\err \to D_\err)) &= \err\\
\end{align*}
Which is monotone because $\bot$ is a maximal element.
The coreflection $(e_\to,p_\to) : D_\err \to D_\err \tl D_\err$ is
defined in the same way, but using the function case of the sum.

